\documentclass[12pt,onecolumn]{IEEEtran}

\makeatletter
\def\ps@headings{%
\def\@oddhead{\mbox{}\scriptsize\rightmark \hfil \thepage}%
\def\@evenhead{\scriptsize\thepage \hfil \leftmark\mbox{}}%
\def\@oddfoot{}%
\def\@evenfoot{}}
\makeatother 
\usepackage{mathrsfs}
\usepackage{amsfonts}
\usepackage{array}       
\usepackage{amssymb}
\usepackage{amsmath}
\usepackage{graphicx}
\usepackage{multirow}   
\usepackage{amsthm}

\usepackage[lined,ruled,commentsnumbered]{algorithm2e}
\usepackage{epstopdf}
\usepackage{amssymb,amsmath}
\usepackage{multirow,url}
\usepackage{color,cite}

\usepackage{amsthm}
\theoremstyle{plain}

\newtheorem{assumption}{Assumption}

\newtheorem{lemma}{Lemma}

\newtheorem{proposition}{Proposition}
\newtheorem{definition}{Definition}
\newtheorem{theorem}{Theorem}

\newtheorem{OA}{Observation}

\ifodd 1
\newcommand{\revh}[1]{{\color{magenta}#1}} 
\newcommand{\com}[1]{\textbf{\color{red} (COMMENT: #1) }} 
\newcommand{\comg}[1]{\textbf{\color{green} (COMMENT: #1)}}
\newcommand{\response}[1]{\textbf{\color{green} (RESPONSE: #1)}} 
\else

\newcommand{\revh}[1]{#1}
\newcommand{\com}[1]{}
\newcommand{\comg}[1]{}
\newcommand{\response}[1]{}
\fi

\newcommand{\footnotesc}[1]{\footnote{#1}}

\def\Ex{\mathrm{E}}

\def\Rset{\mathbb{R}}  
\def\Rsetp{\Rset^2_{+}}   
\def\Probset{\Omega}

\def\N{N}            
\def\Nset{\mathcal{\N}}   
\def\n{n}             

\def\l{s}             

\def\B{\textsf{b}}
\def\LS{\textsf{l}}
\def\A{\textsf{a}}

\def\K{K}                
\def\Kset{\mathcal{\K}}   
\def\k{k}

\def\ch{{channel}}
\def\chs{{channels}}

\def\db{{database}}

\def\lh{{licensee}}

\def\eu{{user}}
\def\eus{{users}}
\def\Eu{{User}}

\def\Nkset{\Nset_{\k}}

\def\RB{B}  
\def\RL{L}  
\def\RA{A}  

\def\Prob{\eta}      
\def\BProb{\boldsymbol{\eta}}      
\def\Proba{\Prob_{\textsc{a}}}      
\def\Probb{\Prob_{\textsc{b}}}      
\def\Probl{\Prob_{\textsc{l}}}      

\def\Probaa{\Prob_{\textsc{a}1}}
\def\Probab{\Prob_{\textsc{a}2}}
\def\Probac{\Prob_{\textsc{a}3}}
\def\Probla{\Prob_{\textsc{l}1}}
\def\Problb{\Prob_{\textsc{l}2}}
\def\Problc{\Prob_{\textsc{l}3}}
\def\BProbaa{\BProb_{1}}      
\def\BProbbb{\BProb_{2}}      
\def\BProbcc{\BProb_{3}}      

\def\fx{f}  
\def\gy{g}  

\def\th{\theta}      
\def\thlb{\th_{\textsc{LB}}}  
\def\thab{\th_{\textsc{AB}}}  
\def\thla{\th_{\textsc{LA}}}  

\def\p{\pi}           

\def\pa{\p_\textsc{a}}           
\def\pl{\p_\textsc{l}}           

\def\w{w}  

\def\U{\Pi}
\def\Ur{\widetilde{\U}}
\def\Ueu{\U^{\textsc{eu}}}
\def\Udb{\U^{\textsc{db}}}
\def\Usl{\U^{\textsc{sl}}}

\def\Udbo{\U^{\textsc{db}}_{0}}
\def\Uslo{\U^{\textsc{sl}}_{0}}

\def\Udbrs{\U^{\textsc{db}}_{\textsc{(i)}}}
\def\Uslrs{\U^{\textsc{sl}}_{\textsc{(i)}}}
\def\Urdbrs{\Ur^{\textsc{db}}_{\textsc{(i)}}}
\def\Urslrs{\Ur^{\textsc{sl}}_{\textsc{(i)}}}

\def\Udbwp{\U^{\textsc{db}}_{\textsc{(ii)}}}
\def\Uslwp{\U^{\textsc{sl}}_{\textsc{(ii)}}}
\def\Urdbwp{\Ur^{\textsc{db}}_{\textsc{(ii)}}}
\def\Urslwp{\Ur^{\textsc{sl}}_{\textsc{(ii)}}}

\def\eq{\triangleq}

\def\R{R}
\def\Ra{\R_{\A}}

\def\InfTV{v}         
\def\InfOut{o}        
\def\InfEU{w}         
\def\InfKnown{\bar{z}}      
\def\InfKnownMin{ \InfKnown_{\textsc{min}} }      
\def\InfUnknown{\hat{z}}    
\def\InfTot{z}        
\def\InfTotA{ \InfTot_{(\A)} }        

\def\pl{\p_\textsc{l}}     

\def\t{t}

\def\ut{U}
\def\rt{\mathcal{R}}

\def\c{c} 
\def\cl{\c_{\textsc{l}}}
\def\ca{\c_{\textsc{a}}}

\def\ch{TV~channel}
\def\chs{TV~channels}

\def\lchs{licensed~\chs}
\def\uchs{unlicensed~\chs}

\def\cs{\c_\textsc{s}}
\def\RS{S} 

\makeatletter
\let\@copyrightspace\relax
\makeatother

\begin{document}

\title{An Integrated Spectrum and Information Market for   Green Cognitive Communications}

 \author{Yuan~Luo, Lin~Gao, and Jianwei~Huang
\IEEEcompsocitemizethanks{
\IEEEcompsocthanksitem
Yuan~Luo, Lin~Gao, and Jianwei~Huang {(corresponding author)} are with Network Communications and Economics Lab (NCEL),
Department of Information Engineering, The Chinese University of Hong Kong, HK,
E-mail: \{yluo, lgao, jwhuang\}@ie.cuhk.edu.hk.
\IEEEcompsocthanksitem
Some of the preliminary results have appeared in \cite{luo2015infocom}.
}
}


\maketitle

\vspace{-18mm}
\begin{abstract}
A database-assisted TV white space network can achieve the goal of green cognitive communication by effectively reducing the energy consumption in cognitive communications.
The success of such a novel network relies on a proper business model that provides substantial incentives for all parties involved.
In this paper, we propose an integrated spectrum and information market for a database-assisted TV white space network, where the geo-location database acts as an online platform providing services to both spectrum market and information market.
We model the interactions among the database operator, the spectrum licensee, and unlicensed users as a three-stage sequential decision process.
Specifically, Stage I characterizes the negotiation between the database and the licensee, in terms of the commission for the licensee to use the spectrum market platform, Stage II models the pricing decisions of the database and the licensee, and Stage III characterizes the subscription behaviors of unlicensed users.
Analyzing such a three-stage model is challenging due to the co-existence of both positive and negative network externalities in the information market.
Nevertheless of this, we are able to explicitly characterize the impact of network externalities on the equilibrium behaviors of all parties involved.
We analytically show that the licensee can never get a market share larger than half in the integrated market.
Our numerical results further show that the proposed integrated market can outperform a pure information market in terms of network profit up to $87\%$.
\end{abstract}


\IEEEpeerreviewmaketitle

%
%
%



\section{Introduction}\label{sec:intro}
\subsection{Background}

With the explosive growth of mobile smartphones and bandwidth-hungry wireless applications,
the corresponding energy consumption due to telecommunication industry is increasing at a unprecedented speed of $16\%-20\%$ per annum \cite{Gur2011}.
Moreover, according to the Climate Group SMART 2020 Report \cite{climatae2008report}, the information and communication technology (ICT) infrastructures account for $3\%$ of global energy consumption and $2\%$ of global CO$_2$ emissions.
Hence, energy optimization of wireless communications, ranging from equipment manufacturing to core functionalities, becomes increasingly important for protecting our environment, coping with global warming, and facilitating sustainable development.

Cognitive communication has been viewed as a promising paradigm for achieving energy-efficient communications.
The key idea is to allow the cognitive radio device to \emph{adapt} its configuration  and transmission decision  to the real-time radio environment.
Hence, a cognitive radio device can select the best reconfiguration operation that balances the energy consumption and communication quality.
Obviously, the success of cognitive communication system greatly relies on the accurate detection of radio environment (\emph{e.g.,} locating the idle channels and figuring out the allowable transmission power to minimize interference to existing users).
If a mobile device is fully responsible for the continuous and accurate detection of radio environment, it would consume a significant amount of energy.
The higher accuracy, the higher computational burden on a mobile device, and thus the higher energy consumption.

In order to reduce energy consumption and guarantee the performance of cognitive communication,
some spectrum regulators (such as FCC in the USA and Ofcom in the UK), together with standards bodies and industrial organizations,\footnote{
Example include IEEE 802.22 WRAN standard (\url{http://www.ieee802.org/22/}) and the industrial companies such as Google (\url{http://www.google.org/spectrum/whitespace/}), Microsoft (\url{http://whitespaces.msresearch.us/}), and Spectrum Bridge (\url{http://www.spectrumbridge.com/}).}
have advocated a \emph{database-assisted} TV white space network architecture.
In such a network, a white space database (called geo-location database) assists unlicensed wireless devices (called white space devices, WSDs) opportunistically exploit the under-utilized UHF/VHF frequency band, which is originally assigned for broadcast television services (hereafter called TV channels) \cite{federal2012third,Ofcom2010geo}.
The main reason for choosing the UHF/VHF frequency band
to support cognitive communications is two-fold.
First, this band is largely under-utilized by the TV broadcast services.
Second, this low-frequency band can support long-distance wireless communications with low transmission power (hence low energy consumption), comparing with the current spectrum band used by cellular and WiFi networks.  ~~~~~~~~~~~~~~~~~~~~~~~~~~~~~~~

In such a database-assisted network, the geo-location database houses a global repository of TV licensees, and updates the licensees' channel occupations periodically.
Each WSD obtains the available TV channel information via querying a geo-location database (via some existing communication networks such as cellular or Wi-Fi networks), rather than having to directly sense the TV channels which lead to significant energy consumption.
In other words,
WSDs are mainly responsible for performing the necessary local computations (\emph{e.g.,} identifying their current locations),
and databases are responsible for performing intensive data processing (\emph{e.g.,} computing the available TV channels for each WSD  based on the channel availabilities and WSD location information).
Such a network architecture can effectively reduce the overall energy consumptions and lead to a green communication ecosystem.


According to existing related regulations,
the geo-location databases are operated by third-party companies  (instead of directly by the regulators or TV license owners).
These database operators such as SpectrumBridge, Microsoft, and Google need to cover its capital expense (CapEx) and operating expense (OpEx)  through a properly designed business model.
Existing models related to the database-assisted network can be categorized into two classes: \emph{Spectrum Market} and \emph{Information Market}.


The spectrum market (\emph{e.g.,} \cite{feng2013database,Bogucka2012,liu2013})
focuses on the trading of
licensed TV channels, which are
registered to some TV licensees but are {under-utilized} by the licensees.
Hence, the licensees can temporarily lease the under-utilized (licensed) {\chs} to WSDs which are able to enjoy an exclusive usage right during a short time period.
This will generate some additional revenue for the licensees.
The database serves as a market platform facilitating such a spectrum market.\footnote{For example, it can act as a spectrum broker or agent, purchase spectrum from licensees and then resell the purchased spectrum to unlicensed users.}
Spectrum Bridge, the world first certified geo-location database, provides such a \emph{database-provided} spectrum market platform called \emph{SpecEx} \cite{SpectrumBridgeCommericial2}.


The information market has been modeled and analyzed for the unlicensed TV channels (\emph{i.e.,} TV white spaces) in our early work \cite{luo2014wiopt, luo2014SDP}.
The {\uchs} are those \emph{not} registered to any TV licensee at a particular location  (for example, outside the official coverage range of the TV towers), hence are the \emph{public resources} at that location.
The spectrum regulators can assign the {\uchs}
for the public and shared usage among unlicensed WSDs, and usually do not allow direct trading of such channels in a spectrum market.
As these channels will be used by WSDs in a shared manner (in contrast to the exclusive usage in the spectrum market),
the communication quality in these {\uchs} is usually not guaranteed.
Notice that the database knows more advanced information regarding the quality of {\uchs} than unlicensed users.\footnote{For example, based on the knowledge about the network infrastructures of TV licensees and their licensed channels, the database can predict the average interference (from licensed devices) on each TV channel at each location.}
Hence, it can sell this information to the unlicensed users through an information market, which not only improves the unlicensed users' expected communication quality, but also provides additional profit to the database.
A commercial example of information market is \emph{White Space Plus} \cite{SpectrumBridgeCommericial}, again operated by Spectrum Bridge.~~~~~~~~~~~~~~~~~~~~~~~~~~~~

All of above listed works considered the spectrum market and information market separately and  independently.
In practice, however, \emph{the licensed TV channels and {\uchs} often co-exist at a particular location}.
Some users may prefer to lease the licensed TV channels for the exclusive usage,
while other users may prefer to share the free {\uchs} with others  (and purchase advanced information if needed).
Hence, a joint formulation and optimization of both spectrum market and information market is important for the practical large scale deployment of the database-based TV white space network.
However, none of the existing work on economics of TV white space networks \cite{feng2013database,Bogucka2012,liu2013} looked at the interaction between spectrum market and information market.
This motivates our study of an \emph{integrated} spectrum and information market for such a database-assisted TV white space network.

\begin{figure}
\centering
 \includegraphics[width=3.1in]{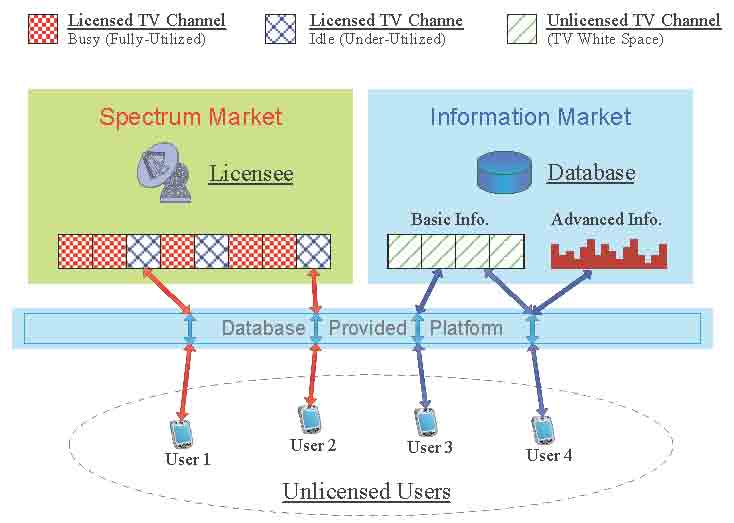}
  \caption{Illustration of database-provided integrated spectrum and information market. Users 1 and 2 lease the licensed TV channels from the spectrum market, and users 3 and 4 share the free {\uchs} with others. User 4 further purchases the advanced information from the information market to improve its performance.}\label{fig:model}
  \vspace{-2mm}
\end{figure}

\subsection{Contributions}

In this paper, we model and study an integrated spectrum and information market for a database-assisted TV white space network,
in which the geo-location database serves as (i) a \emph{spectrum market} platform for the trading of (under-utilized) licensed TV channels between spectrum licensees and unlicensed users\footnote{For convenience, we will call ``white space devices" as ``unlicensed users" in this paper.},
and (ii) an \emph{information market} platform for the trading of information (regarding the {\uchs}) between the database itself and unlicensed users.
Unlicensed users can choose to lease the licensed TV channels from licensees (via the database) for the exclusive usage,
or to share the free {\uchs} with others and purchase the advanced information from the information market if needed.
Figure \ref{fig:model} illustrates such an integrated market.


To understand the market dynamics and equilibrium behaviors in such an integrated market, we formulate the interactions among the geo-location database (operator), the spectrum licensee, and unlicensed users as a three-stage hierarchical model:

\subsubsection{Stage I: Commission Negotiation (Section \ref{sec:layer1})}
In Stage I, the database and the spectrum licensee negotiate regarding the commission for the licensee to use the spectrum market platform.
Specifically, the database provides a platform facilitating the under-utilized licensed TV channels trading between the licensee and unlicensed users. Such a database-assisted licensed spectrum leasing has been widely adopted in the current Licensed Shared Access (LSA) and Authorized Shared Access (ASA) system.
In return, the database takes some \emph{commission charge} from each successful transaction between the spectrum licensee and unlicensed users.
We consider two different commission charging schemes: (i) \emph{revenue sharing} scheme  (RSS), where the licensee shares a fixed percentage of  revenue with the database, and (ii) \emph{wholesale pricing} scheme (WPS), where the database charges a fixed wholesale price from each successful transaction of the licensee.
We assume that both the database and spectrum licensee have market powers, and study the equilibrium commission charge decisions under both schemes using the  Nash bargaining theory \cite{harsanyi1977bargaining}.~~~~~~~

\subsubsection{{Stage II: Price Competition Game} (Section \ref{sec:layer2})}
In Stage II, the database and the spectrum licensee compete with each other for selling information or channels to unlicensed users.
The spectrum licensee decides the price of the licensed TV channels,
and the database decides the price of the advanced information (regarding the unlicensed TV channels).
We analyze such a price competition game, and show that it is a supermodular game \cite{topkis1998supermodular} with nice properties. ~~~~~~~~~

\subsubsection{{Stage III: User Behavior and Market Dynamics (Section \ref{sec:user_subscript})}}
In Stage III, unlicensed users decide the best subscription decisions, given the database's information price and the licensee's spectrum price.
Note that the users need to consider both \emph{negative network externality} (due to congestion and interference) and \emph{positive network externality} (due to the quality of information provided by the database)  of the information market.
We will show how the market dynamically evolves based on users' choices, and what the \emph{market equilibrium} point is.

In our model, Stage I and Stage II focus on cooperation and competition, respectively.
In Stage I,  we study the bargaining process between the database and the licensee, and analyze how they reach an collaborative agreement on the leasing service. In Stage II,  we study the pricing strategies of the database and the licensee who target at different user groups.
The licensee targets at those users who choose the licensed channels for the exclusive usage, while the database targets at those users who choose the unlicensed channels.
The goal of the cooperation is to make the pie larger, and the goal of competition is to decide the way to divide the pie \cite{brandenburger2011co}.

The timescales of   three stages are as follows.
\begin{itemize}
	%
	\item The bargaining in Stage I is performed at a large time scale, e.g., once every year;
	
	\item The price competition in Stage II occurs at a medium time scale, e.g., once every month;
	
	\item The user subscription in Stage III changes at a small time scale, e.g., once every day.
	
\end{itemize}

This three-stage hierarchical order of decision making can be explained as follows. In Stage I, the licensee and the database would not frequently re-negotiate the new terms, as it would consume time and resources to reach an new agreement and change the resource based on the new agreement.
They will only negotiate again when the network resources or network infrastructure change.
In Stage II,
the achieved price equilibrium depends on the database's and the licensee's costs for providing the advanced service or leasing service to users.
In practice, these costs will not change frequently (e.g., change per one month), and neither will the price equilibrium.
In Stage III, a user's subscription decision depends on his instantaneous preference for data rate (i.e., the user's type $\theta$ in our model).
We divide the whole time period into multiple frames, each lasting for certain time (\emph{e.g.,} one day), and assume that the distribution of users' preference remains the same across different time frames.

We summarize the main contributions as follows.~~~~~~~~~~

\begin{itemize}

\item
\emph{Novelty and Practical Significance:}
We propose and study the first integrated spectrum and information market in the literature, for promoting the unlicensed spectrum access to both licensed and unlicensed TV channels. The proposed model captures both the positive and negative network externaltiy of the TV white space network.

\item
\emph{Market Equilibrium Analysis:}
We characterize the sufficient condition under which the proposed integrated market has a unique (user subscription) market equilibrium. We prove that the unique equilibrium is stable, in the sense that a small fluctuation on the equilibrium will drive the market back to the equilibrium.

\item
\emph{Competition Between the Database and the Licensee:}
We study the price competition between the database and the licensee and prove the existence and uniqueness of the price equilibrium. Key contributions of our analysis involves the transformation of the price competition game into an equivalent market share competition game, and the demonstration of the existence and uniqueness of the equilibrium of the transformed game using supermodular game theory.

\item
\emph{Nash Bargaining on the Commission Charge:}
We adopt the Nash bargaining framework to achieve a fair and Pareto-efficient
commission charge between the database and the licensee, under both revenue sharing scheme (RSS) and wholesale pricing scheme (WPS).


\item
\emph{Observations and Insights:}
We show that in this integrated spectrum and information market, the market share equilibrium of the licensee is always no more than half.
In terms of the network profit, our proposed integrated market scheme can improve up to $87\%$ comparing  with a pure information market, and the gap with the coordinated benchmark is less than $10\%$.
When further comparing our two proposed commission charging schemes,
we show that the revenue sharing scheme (RSS) always outperforms wholesale pricing scheme (WPS) in terms of social welfare maximization.
In terms of maximizing database's and licensee's own profit, the result will depend on the level of network externality. When the negative network externality is dominant, RSS is a better choice for the database, while WPS is a better choice fort the licensee. When the positive network externality is dominant, WPS is a better choice for both the database and the licensee.



\end{itemize}


The rest of the paper is organized as follows.
In Section \ref{sec:related}, we review the related literature.
In Section \ref{sec:model}, we present the system model.
In Sections \ref{sec:user_subscript}-\ref{sec:layer1},
we analyze the market equilibrium in Stage III, the price competition game game in Stage II, and the commission bargaining solution in Stage I, respectively.
In Section \ref{sec:simulation}, we provide the simulation results.
Finally, we conclude in Section \ref{sec:con}.


\section{Related Work}\label{sec:related}

Most of the existing studies on green cognitive communications  aimed at addressing the technical issues.
For example,
Hafeez and Elmirghani in \cite{hafeez2015green} presented a new licensed shared access spectrum sharing scheme to increase the energy efficiency in a network.
Palicot in \cite{palicot2009} demonstrated how to achieve green radio communications by employing cognitive radio technology.
Ji \emph{et al.} in \cite{Ji2013} proposed a platform to explore TV white space in order to achieve green communications in cognitive radio network.
Successful commercialization of new green cognitive technology, however, not only relies on sound engineering, but also depends on the proper design of a business model that provides sufficient incentives to the involved parties such as spectrum licensees and the network operators.
The joint study of technology and business issues is relatively under explored in the current green cognitive radio literature.

A common approach for studying market price competition is to model and analyze it as a non-cooperative game.
For example, Niyato \emph{et al.} in  \cite{Niyato2009game} proposed an iterative algorithm to achieve the Nash equilibrium in the competitive spectrum trading market.
Min \emph{et al.} in \cite{Min2012game} studied two wireless service providers' pricing competition by considering spectrum heterogeneity.
Zhu \emph{et al.} in \cite{zhu2014game} studied pricing competition among macrocell service providers via a two-stage multi-leader-follow game.
In the above literature, the market is assumed to be associated with the negative network externality or non-externality.
Luo \emph{et al.} in \cite{luo2014SDP} studied the price competition in the information market of TV white space, where the information market is only associated with the {positive} network externality.
In our work, the integrated market is associated with both the positive and negative network externality. Our numerical results show that the database benefits from the positive network externality, while the licensee benefits from the negative network externality.  Furthermore, which commission charging scheme is better for the database or the licensee
depends on what kind of network externality is dominant in the network. This makes our  market analysis quite different with the above works.



\section{System Model}\label{sec:model}


We consider a database-assisted TV white space network, with a \emph{geo-location {\db}} (or \emph{database} for short) and a set of \emph{unlicensed users} (or \emph{users} for short).
There exist some {\uchs}, which can be  used by unlicensed users freely in a shared manner (\emph{{e.g.,}} using CDMA).
Meanwhile, there is a \emph{spectrum licensee}, who owns some licensed channels and wants to lease the under-utilized channels to users for additional revenue.\footnote{In case there are multiple spectrum licensees, we assume that they are coordinated by the single representative. We will study the more general issue of licensee competition in the future work.}
Different from the {\uchs}, the licensed TV channels can be used by users in an exclusive manner (with the permission of the licensee).
Therefore, users can enjoy a better performance (\emph{{e.g.,}} a higher data rate or a lower interference) on the licensed TV channels.


\subsection{Services Offered by the Database}
Motivated by the current regulatory practices and commercial examples \cite{Ofcom2010geo,SpectrumBridgeCommericial2,SpectrumBridgeCommericial}, we assume  that the database provides the following three services to the users.

\subsubsection{Basic Service}
Regulators such as Ofcom in UK and FCC in US require a {\db} to provide an unlicensed user with the following information \cite{federal2012third,Ofcom2010geo}: (i) the list of {\uchs}, and (ii) each channel's transmission constraints such as a user's maximum allowable transmission power.
Any user can have this \emph{basic (information) service} for free.

\subsubsection{Advanced Service}
Beyond the basic information, the {\db} can also provide certain advanced information regarding the quality of TV channels (as SpectrumBridge did in \cite{SpectrumBridgeCommericial}), as long as it does not conflict with the free basic service. pWe refer to such additional service as the \emph{advanced (information) service}.
Appendix A illustrates an example of the advanced information in terms of the interference level on each channel.
With the advanced information, the {\eu} is able to choose a channel with the  highest quality (\emph{e.g.,} with the lowest interference level).
Hence, the {\db} can \emph{sell} this advanced information to users for profit.
This leads to an \emph{information market}.

\subsubsection{Leasing Service}
As mentioned previously, the database can also serve as a spectrum market \emph{platform} for the trading of licensed channels between the spectrum licensee and users, which we call the \emph{leasing service} (as SpectrumBridge did in \cite{SpectrumBridgeCommericial2}).
In return, the database will charge commission to the spectrum licensee when a trading happens.
Through using the database as the trading platform, the licensees received the \emph{aggregation benefit} \cite{bhargava2004economics}, comparing with the case that they try to directly reach leasing agreement with users.
Specifically, due to the database's proximity to both spectrum licensees and the users, the database's spectrum market platform can aggregate users demand and licensees spectrum, provide trust between participants, and match users and licensees. Hence, the licensees can save time and market efforts in identifying the potential buyers.
We consider two different commission charging schemes: (i) \emph{revenue sharing} scheme (RSS), where the licensee shares a fixed percentage  of  revenue with the database, and (ii) \emph{wholesale pricing} scheme (WPS), where the database charges a fixed wholesale price from each successful transaction, regardless of the licensee's revenue from that transaction.\footnote{Both commission charging schemes are widely used in practice. The revenue sharing scheme is widely used in retail markets such as \cite{cachon2001contracting,gerchak2004revenue,dana2001revenue}. The wholesale pricing scheme is widely used in many newsvendor models such as \cite{Niyato2009game,lariviere2001selling,niyato2008spectrum}.}


\subsection{A User's Choices}

Users can choose either to purchase the licensed channel from the licensee for the exclusive usage, or to share the {\uchs} with others (with and without advanced information).
We assume that all licensed and {\uchs} have the same bandwidth (\emph{e.g.,} 8MHz in UK), and each user only needs one channel (either licensed or unlicensed) at a particular time.

Formally, we denote $\l \in \{\B, \A, \LS\}$ as the  \emph{strategy} of a user, where
\begin{itemize}
\item[(i)] $\l = \B$:
Choose the basic service (\emph{i.e.,} share unlicensed channels with others, without the advanced information);
\item[(ii)] $\l = \A$:
Choose the advanced service (\emph{i.e.,} share unlicensed channels with others, with the advanced information).
\item[(iii)] $\l = \LS$:
Choose the leasing service (\emph{i.e.,} lease the licensed channel for the exclusive usage).
\end{itemize}

We further denote $\RB(\Probl)$, $\RA(\Proba, \Probl)$, and $\RL$ as the {expected} \emph{utilities} that a user can achieve from choosing the basic service ($\l = \B$), the advanced service ($\l = \A$), and the leasing service ($\l = \LS$), respectively. Here, $\Probb$, $\Proba$, and $\Probl$ denote the fractions of {\eus} choosing the basic service, the advanced service, and the leasing service, respectively.
For convenience, we refer to $\Probb$, $\Proba$, and $\Probl$ as the \emph{market shares} of the basic service, the advanced service, and the leasing service, respectively.
Obviously, $\Probb, \Proba, \Probl \geq 0$ and $\Probb + \Proba + \Probl = 1$.
As explained in Section \ref{sec:network_externality}, the values of $\RB(\Probl)$ and $\RA(\Proba, \Probl)$ depend on all users' choices, while the value of $\RL$ is independent of market share.
The \emph{payoff} of a {\eu} is defined as the difference between the achieved utility and the cost (\emph{i.e.,} the  {information price} when choosing the advanced service, or the  {leasing price} if choosing the leasing service).
Let $\th$ denote the \eu's evaluation for the achieved utility\footnote{We consider a simplified but rather general model, where the QoS of a user is a linear function of the expected data rate.}, $\pl \geq 0$ denote the leasing price of the licensee, and $\pa \geq 0$ denote the (advanced) {information price} of the database.
Then, the payoff of a user with an evaluation factor $\th$ is
\begin{equation}\label{eq:utility-basic}
\textstyle
\Ueu_{\th} = \left\{
\begin{aligned}
&\textstyle  \th \cdot \RB(\Probl) ,      &&  \ \text{if} ~ \l = \B, \\
&\textstyle  \th \cdot \RA(\Proba, \Probl)   -  \pa , &&  \  \text{if} ~ \l  = \A, \\
&\textstyle  \th \cdot \RL - \pl ,      &&  \ \text{if} ~ \l = \LS.
\end{aligned}
\right.
\end{equation}



A rational user will choose a strategy $\l \in \{\B, \A, \LS\}$ to maximize its payoff.
Note that {\eus} are heterogeneous in terms of $\th$, which characterizes how different users evaluate the same data rate.
For example, for a user who wants to send a plain text email, his evaluation for a small data rate may be similar to that for a high data rate, as a small data rate is enough for finish sending the email. This can be captured by a small $\th$.
On the other hand, if the email contains a large attachment and needs to be sent within a short amount of time (such as tens of seconds), then the user has a much higher valuation for a high data rate and hence a large $\theta$. Similarly, for a user who is watching a high-quality online video, his evaluation for a small data rate may be much lower than that for a high data rate. This can also be captured by a large $\th$.

Let $\ca$ denote the energy consumption cost of the database for providing the advance service, and let $\cl$ denote the energy consumption cost of the licensee for providing the leasing service. For the rest of the paper, we will also use "operational cost" to refer to these costs.

We further denote $\delta \in [0,1]$ as the revenue sharing percentage under RSS, and $w \geq 0$ as the wholesale price under WPS.
We assume a unit population of agents., (\emph{i.e.,} the total number of {\eus} is equal to $1$).
Then, the  \emph{payoffs} (profits) of the spectrum licensee $\Usl$ and the database $\Udb$, which are defined as the difference between the revenue obtained by providing the services and the cost, are given as follow.
The payoffs under the RSS scheme (Scheme I) are
\begin{equation}\label{eq:u1}
\left\{
\begin{aligned}
\Usl \eq \Uslrs &= ( \pl -\cl)  \Probl   (1 - \delta)
\\
\Udb \eq \Udbrs &= (\pa - \ca)   \Proba + ( \pl - \cl)   \Probl   \delta ,
\end{aligned}
\right.
\end{equation}
and under the WPS scheme (Scheme II) are
\begin{equation}\label{eq:u2}
\left\{
\begin{aligned}
\Usl \eq \Uslwp &= (\pl -\w)  \Probl - \cl  \Probl,
\\
 \Udb \eq \Udbwp &= \pa \Proba + \w   \Probl - \ca  \Proba.
\end{aligned}
\right.
\end{equation}


\subsection{Positive and Negative Network Externalities}\label{sec:network_externality}
Network externalities arise when a {\eu}'s experiencing of consuming a service/product depends on the behavior of other {\eus} in the same network \cite{bhargava2004economics}. In the integrated market that we study, there coexist two types of network externalities.
The \emph{negative} network externality characterizes the user performance degradation due to an increased level of congestion.
The \emph{positive} network externality corresponds to the increasing quality of the (advanced) information as more users purchase the information.
Next we analytically quantify these two network externalities.
As  $\Proba + \Probb  + \Probl  = 1$, sometimes we also denote the total fraction of {\eus} using {\uchs} as $1- \Probl$ in the rest of the paper.

We first have the following observations for a user's expected utility of three strategy choices:
\begin{itemize}
\item
\emph{$\RL$ is a constant independent of $\Proba $, $\Probb$, and $ \Probl $.}
This is because a user can access to the licensed channel exclusively, hence the communication performance on such a channel does not depend on the choices of others.
\item
\emph{$\RB$ is non-increasing in $1 - \Probl $} (the total fraction of users using {\uchs}) due to the congestion effect.
This is because the {\uchs} must be used in a shared manner, hence more users using these channels increases the level of congestion (interference) and reduces the performance of each user. We denote the congestion effect caused by {\eus} using the same {\uchs} as \emph{negative network externality}.
 \item
\emph{$\RA$ is non-increasing in $1 - \Probl$, due to the negative network externality}. This is because the {\uchs} are shared by {\eus}, independent of their choices of subscribing to the {\db}'s advanced service or not.
 \item
\emph{$\RA$ is non-decreasing in $\Proba$}.
As more users subscribe to the database's advanced service, the more information (\emph{e.g.,} users channel choices and transmission power levels) the database knows. In such case, the database can estimate more accurate channel information, which benefits the users who subscribe to the advanced service.
More detailed explanation is provided in Appendix A.
Such benefit that increases with the {\eus} choosing the advanced service is called the \emph{positive network externality}.
\end{itemize}

We write $\RB$ as a non-increasing function of  $ 1- \Probl$, \emph{i.e.,}
\begin{equation}
\label{eq:basic_define}
\textstyle
\RB(\Probl) \triangleq \fx( 1 - \Probl).
\end{equation}

We write $\RA$ as the combination of a
non-increasing function of $ 1 - \Probl$ and a non-decreasing function of $\Proba $, \emph{i.e.,}
\begin{equation}
\label{eq:advanced_define}
\textstyle
\RA(\Proba, \Probl) \triangleq \fx(1 - \Probl) + \gy(\Proba).
\end{equation}
Function $\fx (\cdot)$ reflects the congestion effect, and is identical for $\RB $ and  $\RA$ (as users experience the same congestion effect in both basic and advanced services).
Function $\gy(\cdot)$ reflects the performance gain induced by the advanced information, \emph{i.e.,} the (advanced) information value.

We have the following natural assumption:
\begin{assumption}\label{assume:utility_relationship}
The expected utilities achieved by choosing different services satisfy
$$\RL > \RA(\Proba, \Probl) > \RB(\Probl).$$
\end{assumption}

Comparing with {\uchs}, there is no congestion on the {\lchs}. Hence, the expected utility provided by the leasing service is larger than that provided by the advanced service  (\emph{i.e.,}  $\RL>\RA(\Proba, \Probl)$ and $\RL > \RB(\Probl)$).  As advanced information provides benefit to the users, we have $\RA(\Proba, \Probl) > \RB(\Probl)$.
\footnote{{If we set $\RL < \RB(\Probl)$, no one will choose the leasing service even the leasing service is free of charge. In this case, our integrated model degenerates to the pure information market that is analyzed in \cite{luo2014wiopt}.
If we set $\RA(\Proba, \Probl) = \RB(\Probl)$, then no user will choose the advanced service even the database provides the advanced service for free.
In this case, our integrated model degenerates to a simpler market, where the licensee provides the leasing service and the database provides the basic service only.
The analysis of such a model is simpler than the most general case that we consider here. }}

We further introduce the following assumptions on functions $\fx (\cdot)$ and $\gy (\cdot)$.
\begin{assumption}\label{assum:congestion}
Function $\fx(\cdot)$ is non-negative, non-increasing, convex, and continuously differentiable.
\end{assumption}
\begin{assumption}\label{assum:positive}
Function $\gy(\cdot)$ is non-negative, non-decreasing, concave, and continuously differentiable.
\end{assumption}

Due to the increasing marginal performance degradation under congestion, we assume that function $\fx(\cdot)$ is non-increasing and convex. Such assumption is widely used to model the network congestion effect in wireless networks (see, \emph{e.g.,} \cite{shetty2010congestion,johari2010congestion} and references therein).
Because of the diminishing marginal performance improvement induced by the advanced information, we assume that function $\gy(\cdot)$ is non-decreasing and concave.
Note that the above generic functions $\fx(\cdot)$ and $\gy(\cdot)$ can potentially model a wide range of scenarios where advanced information has different meanings.
We will provide more detailed discussions in the Appendix A.


\subsection{Three-Stage Interaction Model}

%

Based on the above discussion, an integrated spectrum and information market involves the interactions among the database, the spectrum licensee, and the users.
Hence, we formulate the interactions as a three-stage hierarchical model as shown in Figure \ref{fig:layer}.

\begin{figure}[t]
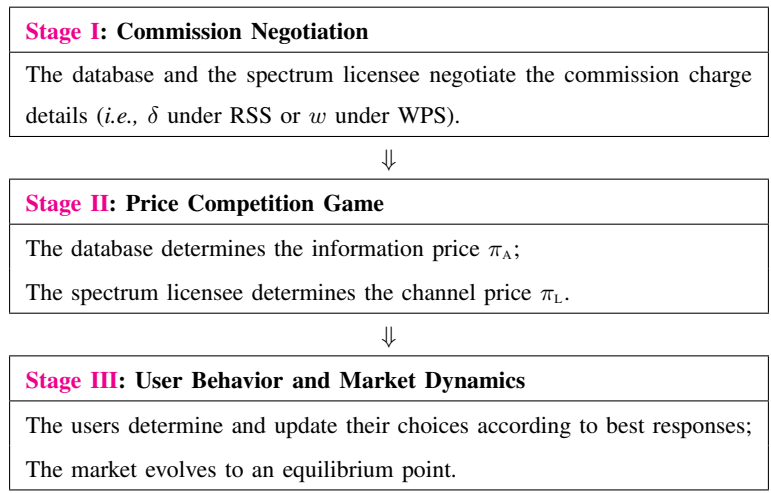

	\centering
	\footnotesize
	\begin{tabular}{|m{3.8in}|}
		\hline
		\textbf{\revh{Stage I}: Commission Negotiation}
		\\
		\hline
		The database and the spectrum licensee negotiate the commission charge details (\emph{i.e.,}
		$\delta $ under RSS or $w$ under WPS).
		\\
		\hline
		\multicolumn{1}{c}{$\Downarrow$} \\
		\hline
		\textbf{\revh{Stage II}: Price Competition Game}
		\\
		\hline
		The database determines the information price $\pa$;
		\\
		The spectrum licensee determines the channel price $\pl$.
		\\
		\hline
		\multicolumn{1}{c}{$\Downarrow$} \\
		\hline
		\textbf{\revh{Stage III}: User Behavior and Market Dynamics}
		\\
		\hline
		The users determine and update their choices according to best responses;
		\\
		The market evolves to an equilibrium point.
		\\
		\hline
	\end{tabular}
	\caption{Three-stage Interaction Model}
	\label{fig:layer}
\end{figure}

Specifically, Stage I captures the negotiation process between the database and the spectrum  licensee, who negotiate the commission charge details of the spectrum market platform, \emph{i.e.,} the revenue sharing factor $\delta $ under RSS, or the wholesale price $\w$ under WPS.
Stage II studies the price competition between the database and the spectrum licensee, where the database determines the advanced information price $\pa$, and the spectrum licensee determines the leasing licensed channel price $\pl$.
Stage III focuses on the subscription behaviors of users, where each user makes his best choice, and dynamically updates the subscription choice based on the current market shares.

In the following sections, we will use backward induction to analyze this three-stage interaction model.


\section{Stage III -- User Behavior and Market Equilibrium}
\label{sec:user_subscript}

In this section, we study the user behavior and market dynamics in Stage III, given the database's information price $\pa$ and the licensee's channel price $\pl$ (in Stage II).
In the following, we first discuss the {\eu}'s best response choice, then show how the user behavior dynamically evolves, finally discuss how the market converges to an equilibrium point.

\subsection{{\Eu}'s Best Response}\label{sec:user-best}

Equations \eqref{eq:utility-basic}, \eqref{eq:basic_define} and \eqref{eq:advanced_define} show that users' choices of services depend on the current market shares of different services. Hence given the market shares, users can compute their best response choices, which in turn will affect the market shares. Next we will characterize such a process in details.


For convenience, we introduce a virtual time-discrete system with slots $\t=1,2,\ldots$, where {\eu}s change their decisions at the beginning of every slot, based on the market shares at the end of the previous time slot.
\footnote{
	The main purpose of
	introducing the virtual time-discrete system, similar as the best response iterative algorithm in classic game theoretic analysis,
	is to characterize the relation between the price and the market equilibrium, and to facilitate the calculation of
	database's optimal price strategy later. Such an analysis technique (i.e., using a dynamic system to understand the outcome of a one-shot system)
	has been extensively adopted in the existing literature, e.g., \cite{manshaei2008evolution,shaolei2011}.}
Let $( \Probl^{\t}, \Proba^{\t}, \Probb^{\t} )$ denote the market shares at the end of slot $t$ satisfying $(\Probl^{\t}, \Proba^{\t}) \in \Probset $, where $\Probset$ is the market share feasible set defined as $\Probset = \{ ( \Probl, \Proba) \in \Rsetp | \Probl + \Proba \leq 1 \}$.
For convenience, we assume that $\th$ is uniformly distributed in $[0,1]$ for all \eus.\footnote{Uniform assumption has been frequently used in the past networking literature (\emph{e.g.,} \cite{manshaei2008evolution,shetty2010congestion,duan2011investment}), and the relaxation to more general distribution often does not change the main insights.}
As each {\eu} will choose a strategy that maximizes its payoff defined in (\ref{eq:utility-basic}),
a type-$\th$ {\eu}'s best response is
\begin{equation}\label{eq:utility_function}
\left\{
\begin{aligned}
\l_{\th}^* = \LS, & \mbox{~~~~iff~~} \th \cdot \RL -  \pl > \max\{ \th \cdot \RA(\Proba^{\t},\Probl^{\t})  -  \pa ,\ \th \cdot \RB(\Probl^{\t}) \},
\\
\l_{\th}^* = \A, & \mbox{~~~~iff~~} \th \cdot \RA(\Proba^{\t},\Probl^{\t})  -  \pa >  \max\{ \th \cdot \RL  -  \pl ,\
 \th \cdot \RB(\Probl^{\t}) \},
\\
\l_{\th}^* = \B, & \mbox{~~~~iff~~} \th \cdot \RB(\Probl^{\t}) > \max\{ \th \cdot \RL  -  \pl, \ \th \cdot \RA(\Proba^{\t},\Probl^{\t}) -  \pa \},
\end{aligned}
\right.
\end{equation}
where $\RB(\Probl^{\t}) =  \fx(1- \Probl^{\t})$ and $\RA(\Proba^{\t},\Probl^{\t}) = \fx(1- \Probl^{\t}) + \gy(\Proba^{\t})$.
\footnote{Here we have written $\Proba^0 + \Probb^0$ as $1-\Probl^0$. For convenience, we will use $\Probl $ (the leasing service's market share) and $\Proba$ (the advanced service's market share) to represent the market state, since the basic service's market share $\Probb = 1 - \Probl - \Proba$ can be directly computed with $\Probl $ and $ \Proba $.}

To better illustrate the above best response, we introduce the following notations:
\begin{equation}\label{eq:p-thres}
\thlb^{\t} \eq \frac{ \pl}{ \RL- \RB(\Probl) },
~~~~
\thab^{\t} \eq \frac{ \pa}{ \RA(\Proba,\Probl) - \RB(\Probl) },
~~~~
\thla^{\t} \eq \frac{\pl-\pa}{\RL - \RA(\Proba,\Probl)}.
\end{equation}

The above three notations denote three user type thresholds.
For example, $\thlb^{\t}$ is defined as the type of user who is indifferent between choosing the leasing service and the basic service (i.e.,  both services provide the user with the same expected payoff).
Hence, a user with a type $\th > \thlb^{\t}$ would achieve a higher expected payoff when choosing the leasing service than choosing the basic service.
Similarly, a user with the type threshold $\thab^{\t}$ is indifferent between the basic service and advanced service, and a user with the type threshold $\thla^{\t}$ is indifferent between the leasing service and advanced service.
Combining these three user types thresholds together, we can compute the market share of each service.


\begin{figure}
\centering
  \includegraphics[width=3in]{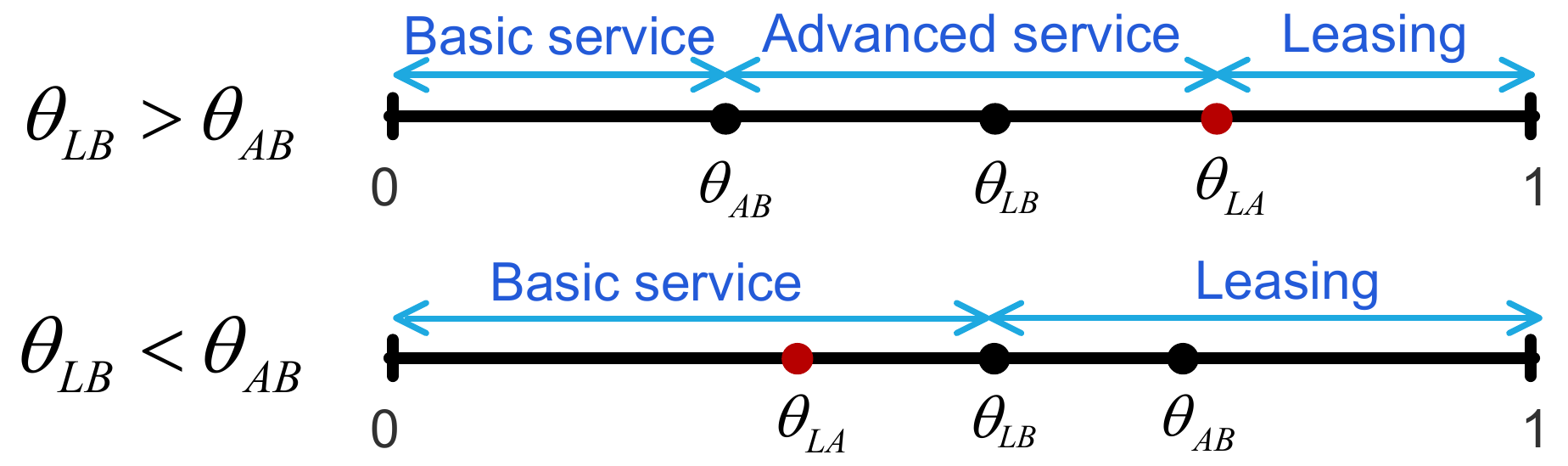}
  \caption{Illustration of $\thlb$, $\thab$, and $\thla$.}\label{fig:threshold}
\end{figure}


Figure \ref{fig:threshold} illustrates several possible relationships among $\thlb^{\t}$, $\thab^{\t}$, and $\thla^{\t}$.
Intuitively,
the users with low values of $\th$ are more willing to choose the free basic service.
The users with medium values of $\th$ are willing to choose the advanced service, in order to achieve a relatively large utility with a relatively low service cost.
The users with high values of $\th$ are more willing to choose the leasing service so that they can obtain a large utility.
Notice that we have $\thlb^{\t} < \thab^{\t}$ if the information price $\pa$ is high or the
information value $\RA(\Proba^{\t},\Probl^{\t})-\RB(\Probl^{\t})$ is low. In this case, users will not choose the advanced service, as shown in the bottom subfugure in Figure \ref{fig:threshold}.



Next we characterize the market shares in slot $\t + 1$ due to users' best responses.
This can help us understand the user behaviour dynamics and market evolutions in the next subsection.
As we assume that $\th$ is uniformly distributed in $[0,1]$, the newly derived market shares $\{ \Probl^{\t+1},\Proba^{\t+1} \}$ in slot $\t+1$, given any market shares $\{\Probl^{\t}, \Proba^{\t} \}$ at the end of $\t$, are
\begin{itemize}
\item
If $\thlb^{\t} > \thab^{\t}$, then $\Probl^{\t+1} = 1 -  \thla^{\t}$ and $\Proba^{\t+1}  = \thla^{\t} - \thab^{\t}$;
\item
If $\thlb^{\t} \leq \thab^{\t}$, then $\Probl^{\t+1} = 1 - \thlb^{\t}$ and $\Proba^{\t+1}  = 0$.
\end{itemize}

We summarize this in Lemma \ref{lemma:market-share}.
\begin{lemma}\label{lemma:market-share}
Given any pair of market shares $( \Probl^{\t},\Proba^{\t} ) \in \Probset$ at the end of slot $\t$, the derived pair of market shares $( \Probl^{\t+1},\Proba^{\t+1})\in\Probset$ in slot $\t+1$ are given by
\begin{equation}\label{eq:user-prob-1}
\textstyle
\left\{
\begin{aligned}
\Probl^{\t+1}  & \textstyle
= \max \big\{ 1 - \max\{ \thla^{\t}, \thlb^{\t} \} ,\ 0  \big\},\\
\Proba^{\t+1}  & \textstyle
= \max \big\{ \min\{ \thla^{\t} ,1 \} - \thab^{\t} ,\ 0  \big\},
\end{aligned}
\right.
\end{equation}
where $\thlb^{\t}$, $\thab^{\t}$, and $\thla^{\t}$ are given in \eqref{eq:p-thres}.
\end{lemma}


Under the assumption that all \eus~update the best responses once synchronously, we can get the results in Lemma \ref{lemma:market-share}.
The detailed proof can be found in Appendix \ref{lemma:market-share-proof}.
Since $\thlb^{\t}$, $\thab^{\t}$, and $\thla^{\t}$ are functions of market shares $\{\Probl^{\t}, \Proba^{\t}\}$, the derived market shares $\{\Probl^{\t+1} ,\Proba^{\t+1} \}$ in slot $\t+1$  are also functions of  $\{\Probl^{\t},\Proba^{\t} \}$, and hence can be written as  $\Probl^{\t+1}(\Probl^{\t}, \Proba^{\t})$ and $\Proba^{\t+1}(\Probl^{\t}, \Proba^{\t})$.

\subsection{{Market Dynamics and Equilibrium}}

When the market shares change,
the users' payoffs (on the advanced service and basic service) change accordingly, as both $\RA(\Proba^{\t},\Probl^{\t})$ and $\RB(\Probl^{\t})$ change.
As a result, users will update their best responses repeatedly, hence the market shares will evolve dynamically until reaching a {stable} point (called \emph{market equilibrium}).
In this subsection, we will study such a market dynamics and equilibrium, given fixed prices $\{ \pl, \pa \}$  (which are decided in Stage II).~~~~~~~~~~~~~~~~~~~


Base on the analysis in Section \ref{sec:user-best}, let $( \Probl^{0}, \Proba^{0} ) \in \Probset$ denote the \emph{initial market shares} in slot $\t = 0$ and $( \Probl^{\t}, \Proba^{\t} ) \in \Probset$ denote the market shares derived at the end of slot $t$ (which serve as the initial market shares for the next slot $t+1$).
We further denote $\triangle \Probl  $ and $\triangle \Proba $ as the changes (dynamics) of market shares between two successive time slots, \emph{e.g.,} $\t$ and $\t+1$, that is,
\begin{equation}\label{eq:user-prob-diff}
\textstyle
\left\{
\begin{aligned}
\triangle \Probl(\Probl^{\t}, \Proba^{\t}) & \textstyle
= \Probl^{\t+1} - \Probl^{\t },\\
\triangle \Proba(\Probl^{\t}, \Proba^{\t}) & \textstyle
= \Proba^{\t+1} - \Proba^{\t },
\end{aligned}
\right.
\end{equation}
where $( \Probl^{\t+1}, \Proba^{\t+1} ) \in \Probset$ are the derived market share in slot $\t+1$, and can be computed by Lemma \ref{lemma:market-share}.

Obviously, if both $\triangle \Probl $ and $\triangle \Proba $ are zero in  a slot $\t+1$, \emph{i.e.,} $\Probl^{\t+1} = \Probl^{\t} $ and $\Proba^{\t+1} = \Proba^{\t} $, then users will no longer change their strategies in the future. This implies that the market achieves the \emph{market equilibrium}.
Formally,

\begin{definition}[Market Equilibrium]\label{def:stable-pt}
A pair of market shares $\BProb^{*} = \{ \Probl^{*}, \Proba^{*} \}\in\Probset$ is a market equilibrium, if and only if
\begin{equation} \label{eq:market_equilibrium}
\textstyle
\triangle \Probl (\Probl^*, \Proba^*) = 0, \mbox{~~~and~~~}
\triangle \Proba (\Probl^*, \Proba^*) = 0.
\end{equation}
\end{definition}

Next, we study the existence and uniqueness of the market equilibrium, and further characterize the market equilibrium analytically.
These results will help us analyze the price competition game in Stage II (Section \ref{sec:layer2}).
\begin{proposition}[Existence and Uniqueness]\label{lemma:uniqueness-eq_pt}
Given any feasible price pair $( \pl, \pa)$,
there exists at least one market equilibrium $(\Probl^*, \Proba^*) \in \Probset$. Furthermore, the market equilibrium is unique and the market dynamics globally converges to it starting from any initial point $\{ \Probl^{0} ,\Proba^{0} \}\in\Probset$ if
\begin{equation}\label{eq:stable_condition}
\max_{(\Probl, \Proba)\in\Probset} \frac{ \gy^{\prime}(\Proba) }{ \gy(\Proba)  } \cdot \frac{ \RL - \RB(\Probl) }{ \RL - \RA(\Proba, \Probl) }  \leq \frac{1}{\kappa},
\end{equation}
where $\kappa = \max_{(\Probl, \Proba)\in\Probset} \{ \frac{\pl - \pa}{ \RL - \RA(\Proba, \Probl)} , \frac{\pa}{\RA(\Proba, \Probl) - \RB(\Probl)} \} $, and $\gy^{\prime}(\Proba)$ is the first-order derivative of $\gy(\cdot)$ with respect to $\Proba$.
\end{proposition}

We prove the convergence of this market dynamics based on the contraction mapping theorem \cite{bertsekas1989parallel},
with the detailed proof in Appendix \ref{lemma:uniqueness-eq_pt-proof}.
A practical implication of condition (\ref{eq:stable_condition}) is that the existence of a unique equilibrium requires the information value $\gy(\Proba)$ (which corresponds to positive network externality) increases slowly with $\Proba$.
Note that the condition (\ref{eq:stable_condition}) is sufficient but not necessary for the uniqueness. Our numerical simulations show that the market converges to a unique equilibrium for a wide range of prices, even when the condition (\ref{eq:stable_condition}) is violated.
Nevertheless, the sufficient condition in (\ref{eq:stable_condition}) suggests that there could be multiple equilibrium points if the impact of positive network externality is significant.

Suppose the uniqueness condition (\ref{eq:stable_condition})  is satisfied. We characterize the unique equilibrium by the following theorem.
\begin{theorem}[Market Equilibrium]\label{thrm:stable-eq_pt}
Suppose the uniqueness condition (\ref{eq:stable_condition}) holds.
Then, for any feasible price pair $( \pl, \pa)$,
\begin{itemize}
\item[(a)]
If $ \left.\thlb(\Probl, \Proba)\right|_{\Probl = 0} \leq \left.\thab(\Probl, \Proba)\right|_{\Proba = 0}$, there is a unique market equilibrium  $\BProb^{\dag} = \{ \Probl^{\dag}, \Proba^{\dag} \}$ satisfies
    \begin{equation}\label{eq:NE-pt-1}
 \textstyle  \Probl^{\dag} = 1 - \thlb(\Probl^{\dag},\Proba^{\dag}) , \mbox{~~~and~~~}
  \Proba^{\dag} = 0;
    \end{equation}

\item[(b)]
If $ \left.\thlb(\Probl, \Proba)\right|_{\Probl = 0} > \left.\thab(\Probl, \Proba)\right|_{\Proba = 0}$, there is a unique market equilibrium  $\BProb^{*} = \{ \Probl^{*}, \Proba^{*} \}$ satisfies
    \begin{equation}\label{eq:NE-pt-22}
    \left\{
      \begin{aligned}
      &\textstyle  \Probl^{*} = 1 - \thla(\Probl^{*}, \Proba^{*}) , \\
      &\textstyle  \Proba^{*} = \thla(\Probl^{*}, \Proba^{*}) - \thab(\Probl^{*}, \Proba^{*}).
       \end{aligned}
    \right.
    \end{equation}
\end{itemize}
\end{theorem}

\begin{proof}
First, we obtain the derived market shares by substituting the market shares given in (\ref{eq:NE-pt-1}) or (\ref{eq:NE-pt-22}) into (\ref{eq:user-prob-1}).
Then, we can check that the above derived market shares satisfy the equilibrium condition (\ref{eq:market_equilibrium}).
For the detailed proof, please refer to Appendix \ref{thrm:stable-eq_pt-proof}.
\end{proof}

A practical implication of Theorem 1 is that the information price $\pa$ should not be too high or the information value (i.e., $\RA-\RB$) should be large enough; otherwise, no users will choose the advance service at the equilibrium.




\section{Stage II -- Price Competition Game Equilibrium}\label{sec:layer2}

In this section, we study the price competition between the database and the spectrum licensee in Stage II, given the commission negotiation solution in Stage I. Based on the analysis of Stage III in Section IV, the database and spectrum licensee are able to predict the user behavior and market equilibrium in Stage III  when making their pricing decisions.
We will analyze the  pricing equilibrium under both the revenue sharing scheme (RSS) and the wholesale price scheme (WPS).


\begin{definition}[Price Competition Game]
	The Price Competition Game (PCG) is defined as follows.
	\begin{itemize}
		\item
		\emph{Players:} The database and the spectrum licensee;
		\item
		\emph{Strategies:} Information price  $\pa \geq 0$ for the database, and $\pl \geq 0$ for the licensee;
		\item
		\emph{Payoffs:}  Payoff is defined in (\ref{eq:u1}) under RSS, and in (\ref{eq:u2}) under WPS.
	\end{itemize}
\end{definition}

For the rest of this section, we assume that condition (8) holds. Then, we write the unique market equilibrium $\BProb^{*} = \{ \Probl^{*}, \Proba^{*} \}$ in Stage III as functions of prices $(\pl, \pa)$, \emph{i.e.,} $\Probl^{*} (\pl, \pa)$ and $\Proba^{*} (\pl, \pa)$.
Intuitively, we can interpret $\Probl^{*}$ and $\Proba^{*}$ as the \emph{demand} functions of the licensee and the \db, respectively.

\subsection{Revenue Sharing Scheme --- RSS}\label{sec:layer2-rss}

We fist study the game equilibrium under RSS, where the licensee shares a fixed percentage $\delta \in [0, 1]$ of revenue with the database.
By (\ref{eq:u1}), the payoffs of the \lh~and the \db~can be written as:
\begin{equation}\label{eq:sl-profit-dynamic-rv}
\left\{
\begin{aligned}
\Uslrs(\pl , \pa) & = ( \pl - \cl) \cdot \Probl^{*}(\pl, \pa) \cdot (1 - \delta),
\\
\Udbrs(\pl , \pa) &= ( \pa - \ca )  \Proba^{*}(\pl, \pa) + ( \pl  - \cl)  \Probl^{*}(\pl, \pa)  \delta.
\end{aligned}
\right.
\end{equation}


\begin{definition}[Price Equilibrium]\label{def:nash}
A price pair  $( \pl^{*}, \pa^{*} )$ is a Nash equilibrium, if
\begin{equation}\label{eq:db-price-dynamic}
\left\{
\begin{aligned}
\textstyle\pl^{*} & = \arg \max_{\pl \geq 0}\ \Uslrs(\pl , \pa^{*}),
\\
\textstyle\pa^{*} & = \arg \max_{\pa \geq 0}\ \Udbrs(\pl^{*} , \pa).
\end{aligned}
\right.
\end{equation}
\end{definition}


It is difficult to analytically characterize the market equilibrium $\{\Probl^{*}(\pl, \pa),\Proba^{*} (\pl, \pa) \}$ under a particular price pair $\{\pl, \pa\}$.
We tackle this challenge by transforming the original price competition game (PCG) into an
equivalent \emph{market~share~competition game} (MSCG).
In such a case, the market shares are the strategies of the database and the licensee, while the prices are the functions of the market shares.


A key observation of such a transformation is that, under the uniqueness condition (\ref{eq:stable_condition}), there is a \emph{one-to-one} correspondence between the market equilibrium $\{\Probl^{*}, \Proba^* \}$ and the prices  $\{\pl, \pa\}$.
Because of this, once the \lh~and the \db~choose the prices $\{\pl, \pa\}$, they have equivalently chosen the market shares $\{\Probl^{*}, \Proba^* \}$.
Substitute $\thla=\frac{\pl-\pa}{\RL - \RA}$ and $\thab=\frac{ \pa}{ \RA- \RB }$ into (\ref{eq:NE-pt-22}), we can derive the inverse function of (\ref{eq:NE-pt-22}), where prices are functions of market shares defined on $\Probset$, \emph{i.e.,}\footnote{Note that we omit the trivial case in (\ref{eq:NE-pt-1}),
where the database has a zero market share in terms of the advanced service.
In this case, the licensee can determine the market share splitting (between leasing service and basic service from the database) by optimizing his leasing price $\pl$.
}
\begin{equation}\label{eq:price-market-share-rs}
\textstyle
\left\{
  \begin{aligned}
  \textstyle  \pl(\Probl , \Proba )   = & ( 1 - \Probl ) \cdot \left( \RL - \fx(1-\Probl) - \gy(\Proba) \right)   + ( 1 - \Probl - \Proba )\cdot \gy(\Proba) , \\
 \textstyle  \pa (\Probl , \Proba ) =& ( 1 - \Probl - \Proba )\cdot \gy(\Proba).
   \end{aligned}
   \right.
\end{equation}
The payoffs of the database and the licensee can be rewritten as:
\begin{equation}\label{eq:sl-profit-dynamic-rv-xx}
\left\{
\begin{aligned}
\textstyle
\Urslrs(\Probl , \Proba) & = ( \pl(\Probl, \Proba) - \cl )  \cdot \Probl \cdot (1 - \delta) ,
\\
\textstyle
\Urdbrs(\Probl , \Proba) & = ( \pa(\Probl , \Proba) -\ca ) \Proba + ( \pl (\Probl , \Proba) - \cl)  \Probl  \delta.
\end{aligned}
\right.
\end{equation}

\begin{definition}[Market Share Competition Game]
	The equivalent Market Share Competition Game (MSCG) is defined as follows.
	\begin{itemize}
		\item
		\emph{Players:} The database and the spectrum licensee;
		\item
		\emph{Strategies:} Market share  $\Proba$ for the database, and $\Probl$ for the licensee, where $(\Probl, \Proba) \in \Probset$;
		\item
		\emph{Payoffs:} Payoffs are defined in (\ref{eq:sl-profit-dynamic-rv-xx}).
	\end{itemize}
\end{definition}


\begin{definition}[Market Share Equilibrium]\label{def:nash_mscg}
Market shares $(\Probl^*, \Proba^*)\in\Probset$ is a Market Share Equilibrium if
\begin{equation}\label{eq:db-price-dynamic}
\textstyle \Probl^* = \arg \max_{\Probl} \Uslrs(\Probl , \Proba^*),
~~\text{and}~~\Proba^* = \arg \max_{\Proba} \Udbrs(\Probl^*, \Proba).
\end{equation}
\end{definition}


%
%
%
We first show that the equivalence between the original PCG and the above MSCG.
\begin{proposition}[Equivalent Games]\label{lemma:game_tranform}
If $\{\Probl^{*}, \Proba^*\}$ is a Market Share Equilibrium of MSCG, then $ \{\pl^*, \pa^*\}$ given by (\ref{eq:price-market-share-rs}) is a Price Equilibrium of PCG.
\end{proposition}



We next characterzie the Market Share Equilibrium of the MSCG. We first give the following proposition which bounds the market shares maximizing the database's and the licensee's expected payoffs in \eqref{eq:db-price-dynamic}.

\begin{proposition}[Boundary of Market Share Equilibrium]\label{lemma:market_share_boundary}
	For any $\{\Probl^{*}, \Proba^*\}$ that is a solution of \eqref{eq:db-price-dynamic}, we have
$\Probl^*\in (0,1/2)$ and $\Probl^* + \Proba^* < 1$.
\end{proposition}

Proposition \ref{lemma:market_share_boundary} shows
that the licensee will achieve an equilibrium market share that is smaller than half.
Intuitively, the objective of the licensee is to maximize its own profit,
and a larger market share of licensee corresponds to a smaller market price, which does not necessarily increase the licensee's profit.
Another key insight of Proposition \ref{lemma:market_share_boundary} is that considering
\eqref{eq:db-price-dynamic} alone is enough to guarantee the feasibility constraint of $(\Probl, \Proba)\in\Probset$ in Definition \ref{def:nash_mscg}.
Hence, we can study the existence and uniqueness of the MSCG market share equilibrium by analyzing \eqref{eq:db-price-dynamic} only.

%
%


\begin{theorem}[Existence and Uniqueness of Market Share Equilibrium]\label{thrm:NE-existence}
Given the commission charge $\delta$,
there exists at least one Market Share Equilibrium of MSCG $(\Probl^*, \Proba^*)\in\Probset$. Furthermore, a sufficient condition for the uniqueness of the market share equilibrium is
\begin{equation}\label{eq:unique_NE_price}
\textstyle
 \frac{  \partial^2{ \Urslrs({\Prob}_{\textsc{l}} , \Proba) } }{ \partial{ {\Prob}_{\textsc{l}} }^2 } \leq \frac{  \partial^2{ \Urslrs( {\Prob}_{\textsc{l}} , \Proba)  } }{ \partial{ {{\Prob}_{\textsc{l}} } }\partial{ \Proba  } }
~\text{and}~
\textstyle  \frac{  \partial^2{ \Urdbrs( {\Prob}_{\textsc{l}} , \Proba) } }{ \partial{ {\Prob}_{\textsc{a}} }^2 } \leq \frac{  \partial^2{ \Urdbrs( {\Prob}_{\textsc{l}} , \Proba)  } }{ \partial{ { \Proba } }\partial{ {\Prob}_{\textsc{l}}  } }, ~~\textstyle \forall (  \Probl, \Proba ) \in \Probset.
\end{equation}
\end{theorem}

\begin{proof}[Proof Sketch]
First, we can prove that the MSCG is a supermodular game (with a proper strategy transformation), and thus has at least one Nash Equilibrium. Based on the supermodular game theory \cite{topkis1998supermodular}, the MSCG has a unique Nash Equilibrium as long as it satisfies the sufficient condition given in \eqref{eq:unique_NE_price}.
\end{proof}

For a supermodular game with a unique Nash equilibrium,  several commonly used updating rules are guaranteed to converged to the NE \cite{topkis1998supermodular}. In this paper, we use the best response algorithm as in \cite{huang2006distributed}. Due to the space limit, we provide the detailed algorithm in
Appendix \ref{appendix:best_response}.
	
Once we obtain the Market Share Equilibrium $(\Probl^*, \Proba^*)$ of MSCG, we can compute the Price  Equilibrium $(\pl^*, \pa^*)$ of the original PCG by (\ref{eq:price-market-share-rs})

\subsection{Wholesale Pricing Scheme --- WPS}\label{sec:layer2-wps}

We now study the game equilibrium under WPS, where the database charges a fixed wholesale price $w \geq 0$ from each successful transaction of the licensee.
By (\ref{eq:u2}), the payoffs of the \lh~and the \db~can be written as:
\begin{equation}\label{eq:sl-profit-dynamic-whole}
\left\{
\begin{aligned}
\Uslwp(\pl , \pa) &= ( \pl - \w) \cdot \Probl^{*}(\pl , \pa) - \cl \Probl^{*}(\pl , \pa),
\\
\Udbwp(\pl , \pa) &= ( \pa - \ca ) \cdot \Proba^{*}(\pl , \pa) + \w \cdot \Probl^{*}(\pl , \pa).
\end{aligned}
\right.
\end{equation}


With the similar analysis used in Section \ref{sec:layer2-rss}, we can transform the Price Competiton Game (PCG) into an equivalent Market Share Competition Game (MSCG), and show that the MSCG is a supermodular game.
Our key results about the existence and uniqueness of game equilibrium are as follows.
\begin{theorem}\label{them:uniq_wps}
	Given the wholesale price $\w$, there exists a unique Market Share Equilibrium $\Probl^* \in (0,1/2)$ and $\Probl^*+\Proba^* < 1$
	for the MSCG, and thus a unique Price Equilibrium, if the following conditions are satisfied
	\begin{equation}\label{eq:unique_NE_price_WPS}
	\textstyle
	\frac{  \partial^2{ \Urslwp({\Prob}_{\textsc{l}} , \Proba) } }{ \partial{ {\Prob}_{\textsc{l}} }^2 } \leq \frac{  \partial^2{ \Urslwp( {\Prob}_{\textsc{l}} , \Proba)  } }{ \partial{ {{\Prob}_{\textsc{l}} } }\partial{ \Proba  } }
	~\text{and}~
	\textstyle  \frac{  \partial^2{ \Urdbwp( {\Prob}_{\textsc{l}} , \Proba) } }{ \partial{ {\Prob}_{\textsc{a}} }^2 } \leq \frac{  \partial^2{ \Urdbwp( {\Prob}_{\textsc{l}} , \Proba)  } }{ \partial{ { \Proba } }\partial{ {\Prob}_{\textsc{l}}  } }, ~~\forall (  \Probl, \Proba ) \in \Probset.
	\end{equation}
\end{theorem}


The detailed discussions in Appendix H.

Although we use the similar method to derive the Nash equilibrium of the PCG game under both RSS and WPS, these two equilibria are quite different.
Intuitively, under RSS, the objective of the database is consistent with that of the licensee.
This can be shown by the common term in both players' payoffs given in
(\ref{eq:sl-profit-dynamic-rv}).
Hence, the database becomes less aggressive in competing with spectrum licensee under RSS that under WPS.

To emphasize the fact that the equilibrium payoffs in Stage II depend on $\delta$ (under RSS) or $\w$ (under WPS), we will write
the equilibrium payoff of the database as $\Udbrs (\delta)$ under RSS and $\Udbwp(\w)$ under WPS. Similarly, we will
write the equilibrium payoff of the licensee as $\Uslrs (\delta)$ under RSS and $\Uslwp(\w)$ under WPS.

\section{Stage I -- Commission Bargaining Solution}\label{sec:layer1}



In this section, we study the commission negotiation among the database and the spectrum licensee in Stage I, based on their predictions of the price equilibrium in Stage II and the market equilibrium in Stage III\footnote{In our proposed bargaining model, the database and the licensee only conduct the bargaining if this leads to a payoff increase for both sides.
Otherwise they can choose not to bargain and do not cooperate in Stage I.
Furthermore,
based on the discussions in the previous paragraph, we can see that the licensee does not have the full market power. Hence a bargaining model is suitable for such a market.
}.

Specifically, we want to find a feasible revenue sharing percentage $\delta \in [0, 1]$ under RSS, or a feasible wholesale price $\w \geq 0$ under WPS, that is \emph{satisfactory} for both the database and the spectrum licensee.
This is motivated by the fact that both the database operator (\emph{e.g.,} Google, Microsoft, and SpectrumBridge) and spectrum licensees (\emph{e.g.,} AT$\&$T, Verizon, and China Mobile) have considerable market power, and one side can not determine $\delta$ or $\w$ alone.
We formulate the commission negotiation problem as a \emph{bargaining problem}, and solve it using the Nash bargaining theory \cite{harsanyi1977bargaining}.


\subsection{Revenue Sharing Percentage $\delta$-Bargaining under RSS}\label{sec:layer1-rss}

We first study RSS, where we want to determine the revenue sharing percentage $\delta \in [0, 1]$.

We first derive the database's and the licensee's payoffs when reaching an agreement and when \emph{not} reaching any agreement. This allows us to characterize the payoff improvements due to successful bargaining.

Formally, when reaching an agreement $\delta$, the database's and the licensee's payoffs are  $\Udbrs (\delta)$ and $\Uslrs (\delta)$ derived in
Section \ref{sec:layer2-rss}, respectively.
When not reaching any agreement (reaching the disagreement), the licensee's profit is $\Uslo = 0$, and the database's profit is $\Udbo = \pa^{\dag} \cdot \Proba^{\dag}(\pa^{\dag}) $, where $\pa^{\dag}$ and $\Proba^{\dag}(\pa^{\dag})  $ are the database's optimal price and the corresponding market share in the pure information market.\footnote{Such an optimal price and the corresponding market share can be derived in the same way as in Section \ref{sec:layer2-rss}, by simply setting $\thla = 1$.}
Then, the Nash bargaining solution is the solution of the following optimziation problem,
\begin{equation}\label{eq:NBS-RS}
\begin{aligned}
\max_{\delta \in [0,1]}~ &\left( \Udbrs(\delta) - \Uslo \right) \cdot \left( \Uslrs(\delta) - \Udbo \right) \\
\text{s.~t.~} & \Udbrs(\delta) \geq \Uslo,~~ \Uslrs(\delta) \geq \Udbo.
\end{aligned}
\end{equation}

Analytically solving (\ref{eq:NBS-RS}) may be difficult, if we are not able to analytically characterize $\Udbrs(\delta)$  and $\Uslrs(\delta) $.
Nevertheless, we notice that the bargaining variable $\delta$ lies in a closed and bounded range of $[0,1]$, and the objective function of (\ref{eq:NBS-RS}) is bounded.
Hence, there exists at least one optimal solution for (\ref{eq:NBS-RS}).
As our numerical results show that the obejctive function of \eqref{eq:NBS-RS} is approximately quadratic in $\delta$, the optimal solution is unique and can be found by using  one-dimensional search methods (\emph{e.g.,} \cite{bargai2001search}).

\subsection{Wholesale Price $\w$-Bargaining under WPS}\label{sec:layer1-wps}

We now study WPS, where the database charges the spectrum licensee a fixed wholesale price $\w\geq 0$ for each successful  transaction of the latter.

When reaching an agreement $\w$, the database's and the licensee's payoffs are  $\Udbwp (\w)$ and $\Uslwp (\w)$ derived in
Section \ref{sec:layer2-wps}, respectively.
When not reaching any agreement (reaching the disagreement), the licensee's profit is $\Uslo = 0$, and the database's profit is $\Udbo = \pa^{\dag} \cdot \Proba^{\dag}(\pa^{\dag}) $, which is same as that under RSS (Section \ref{sec:layer1-rss}).
The Nash bargaining problem is
\begin{equation}\label{eq:NBS-WP}
\begin{aligned}
\max_{\w \geq 0}~ &\left( \Udbwp (\w) - \Uslo \right) \cdot \left( \Uslwp (\w) - \Udbo \right) \\
\text{s.~t.~} & \Udbwp (\w)  \geq \Uslo,~~ \Uslwp (\w) \geq \Udbo.
\end{aligned}
\end{equation}

%
%
%
%

We further notice that that we can restrict the bargaining variable $\w$ within a closed and bounded set, say $[0, \RL]$, while not affecting the optimality of the solution.
This is no user will choose the leasing service when $\w > \RL$, in which case the spectrum licensee will get a zero payoff. This is certainly not an optimal solution of \eqref{eq:NBS-WP}.
Similar as Section \ref{sec:layer1-rss},
there would exist a unique solution for (\ref{eq:NBS-WP}) as shown in our simulations, and which can be found effectively through one-dimensional numerical search.



\section{Simulation Result}\label{sec:simulation}
We evaluate the system performance
(\emph{e.g.,} the \db's profit, the network profit, and the social welfare)
achieved under both revenue sharing scheme (RSS) and wholesale price scheme (WPS)  through extensive numerical studies. We will focus on the impact of system parameters (\emph{i.e.,} the degree of network externality and the \lh's operating cost) on system performance.

\subsection{Simulation Setting}\label{sec:simulaton_function}
As a concrete example,
we choose $\fx( \Probl ) = \alpha_1 - {\beta_1}\cdot ( 1 - \Probl )^{\gamma_1}$ to model the negative network externality with the following justifications.
First, when $\Probl = 1$, there is no congestion among users choosing {\uchs}, as all {\eus}  lease {\lchs} for exclusive usage. In this case, the utility provided by the basic service is $\alpha_1$.
Second, when $\Probl = 0$, all {\eus} choose the shared {\uchs}, and the congestion is maximum.
In this case, the utility provided by the basic service decreases to $\alpha_1 - \beta_1$.
Finally, the parameter $\gamma_{1} \in (0,1]$ models the elasticity of the negative network externality.
A small $\gamma_{1}$ means that the value of $\fx(\Probl)$ will be small even with a large $\Probl$. This means that a small fraction of {\eus} using {\uchs} will cause a large enough congestion.

Similarly, we use function $\gy(\Proba) = \alpha_2 + (\beta_2 - \alpha_2) \cdot {\Proba}^{\gamma_2}$ to model the positive network externality.
Specifically, $\alpha_1$ denotes the minimum benefit brought by the database's advanced information on {\uchs} occupation, and $\beta_2$ denotes the maximum benefit brought by the database's advanced information. The parameter $\gamma_2 \in (0, 1]$ characterizes the elasticity of the positive network externality.\footnote{{These two functions generalize the linear network externality models from many existing literatures such as \cite{easley2010effect}, satisfy our Assumptions \ref{assum:congestion} and \ref{assum:positive}, and are used to model the network effect in the literature (see, \emph{e.g.,} \cite{wu2014exploring}).}}



To characterize the interaction of negative and positive network externality in the information market, we first derive the first-order derivative of $\RA$ with respect to $\Proba$, \emph{i.e.,}
$
\textstyle
\frac{ \partial{\RA} }{ \partial{\Proba} } = -{\gamma_1} \cdot{\beta_1}\cdot ( \Proba + \Probb )^{\gamma_1 - 1} + {\gamma_2}\cdot {\beta_2}\cdot ( \Proba )^{\gamma_2-1}.
$
Notice that if $\frac{ \partial{\RA} }{ \partial{\Proba} } > 0$, the value of advance information increases with the market share $\Proba$ (\emph{i.e.,} the percentage of users purchasing the advanced information).
We call this the \emph{positive network externality dominant} case (or simply positive dominant).
On the other hand, if $\frac{ \partial{\RA} }{ \partial{\Proba} } < 0$,
the value of advance information decreases with the market share $\Proba$. We call this the \emph{negative network externality dominant} case (or simply negative dominant).

To facilicate our study, we set $\gamma_1 = \gamma_2$ and change the value of $\beta_1$ and $\beta_2$ to examine different cases of network externality.
Specifically, we use $\lambda = \beta_2/\beta_1$ to represent the degree of network externality.
Obviously, if $\lambda = \beta_2/\beta_1  >  (\frac{\Proba + \Probb}{\Proba})^{1 - \gamma_1}$,
then $\frac{ \partial{\RA} }{ \partial{\Proba} } > 0$ for any $\{ \Proba, \Probb \}$, hence the information market is positive dominant.
If $\lambda < (\frac{\Proba + \Probb}{\Proba})^{1 - \gamma_1}$,
then $\frac{ \partial{\RA} }{ \partial{\Proba} } < 0$, hence the information market is negative dominant.
Hence an increasing value of $\lambda$ implies that the positive network externality is getting relatively stronger.

Unless specified otherwise, we assume that
$\alpha_1 = 1$, $\alpha_2 = 1$, $\beta_1 = 1$, $\gamma_1 = 0.6$, $\gamma_2 = 0.6$, and $\ca = 0.2$  in the rest of the numerical studies.


\subsection{The Impact of Network Effect Parameter $\lambda$}


Figure \ref{fig:performance_vs_gamma} illustrates (a) the \db's profit, (b) the \lh's profit, and (c) the  network profit,  \emph{i.e.,} the aggregate profit of the \db~and the \lh~
achieved under different network effect.
Here we choose
$\lambda$ from $0.4$ to $1.8$, \emph{i.e.,} the network effect changes from strong negative externality to strong positive externality. In this simulation, we fix the \lh's operational cost as $\cl = 0.9$ and the quality of leasing service as $\RL = 6$.

\begin{figure*}
	\centering
	\includegraphics[width=2.2in]{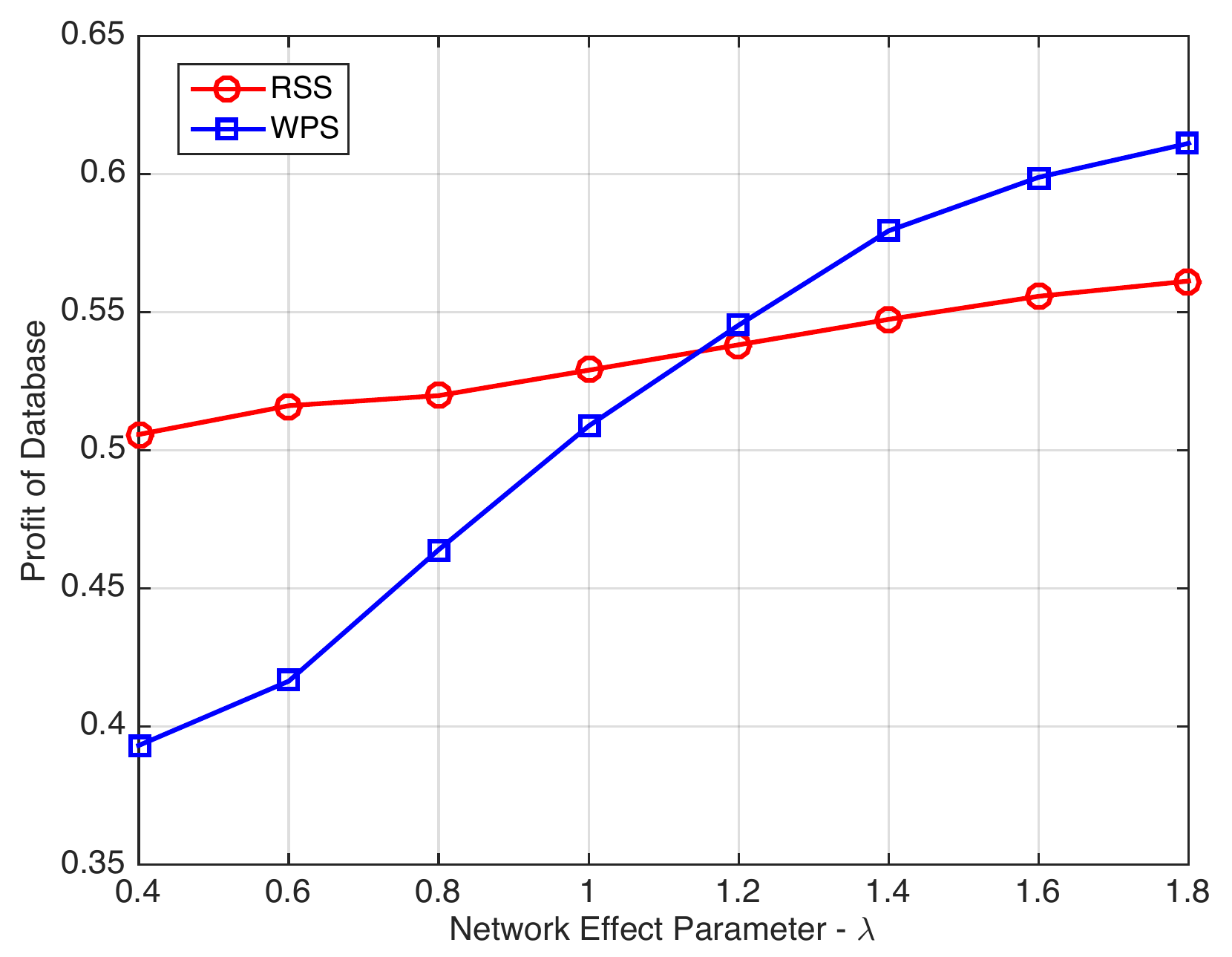}
	\includegraphics[width=2.2in]{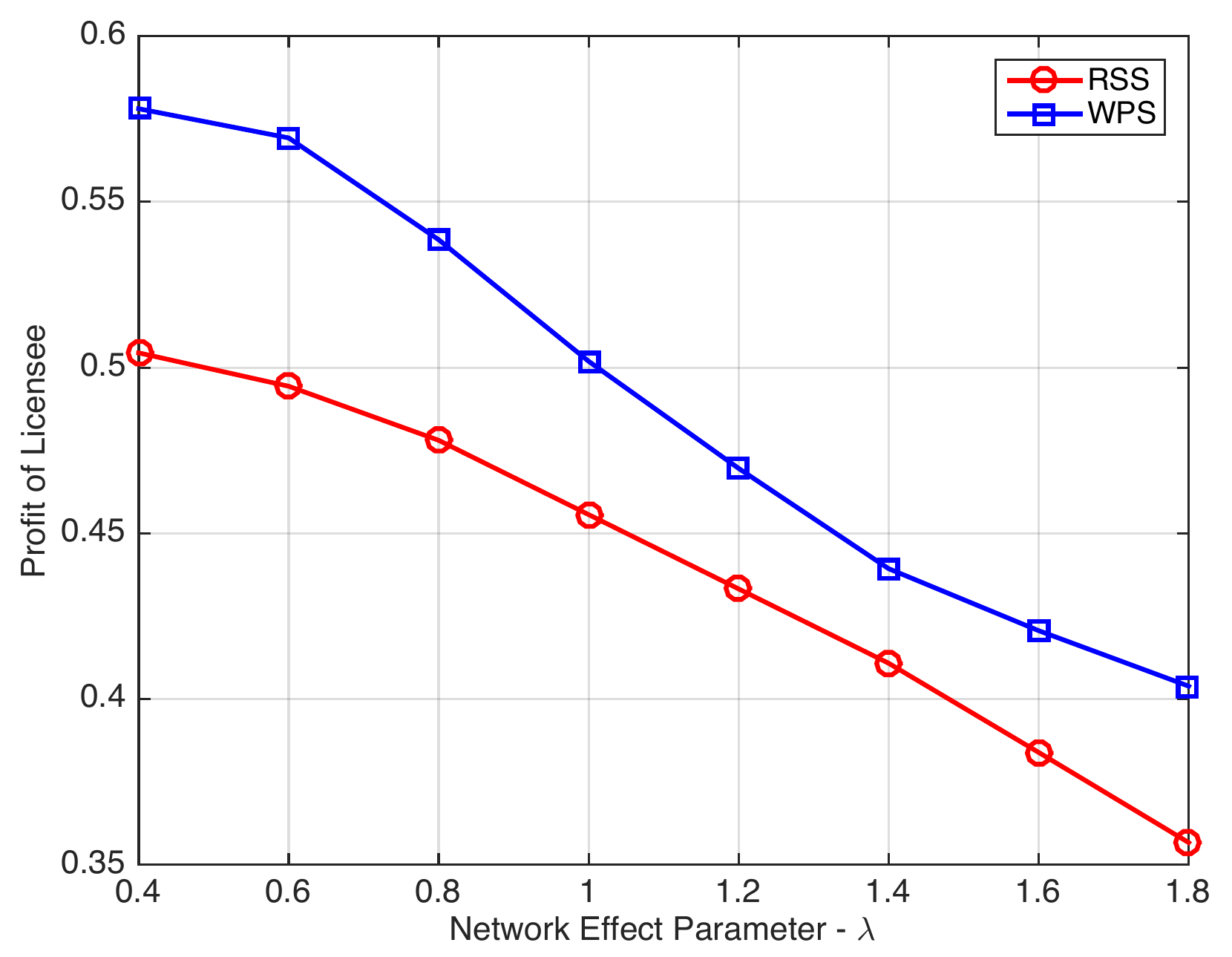}
	\includegraphics[width=2.18in]{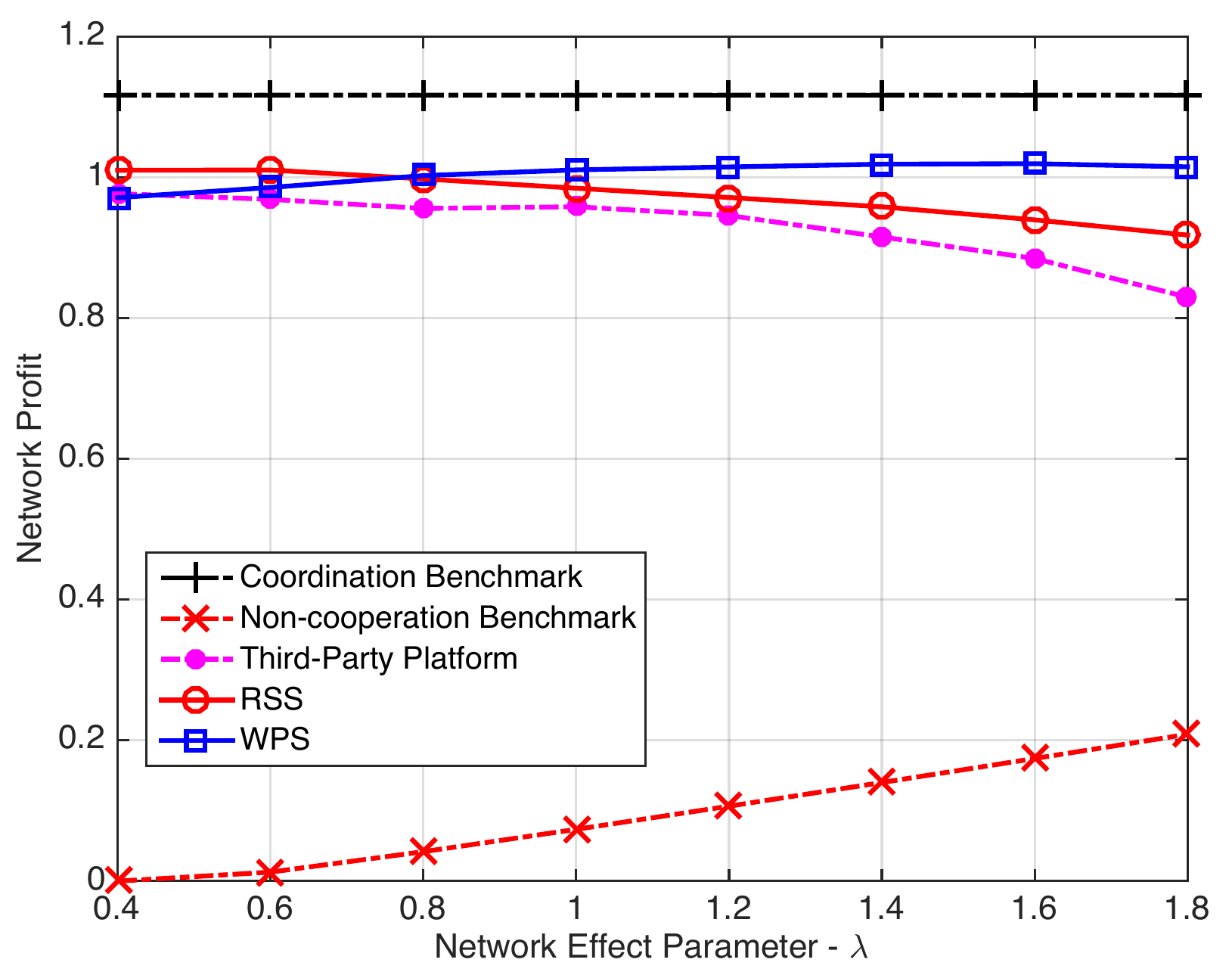}
	\caption{(a) The \db's profit, (b) The \lh's profit, (c) Network profit vs the level of network externality $\lambda$ under the revenue sharing scheme (RSS) and the wholesale price scheme (WPS). }\label{fig:performance_vs_gamma}
\end{figure*}

In Figure \ref{fig:performance_vs_gamma}.c, we use the black dash-dot line (with mark $+$) to denote the \emph{coordination benchmark}, where the \db~and the \lh~act as an integrated party to maximize their aggregate profit. We use the red dash-dot line (with mark $\times$) to denote the \emph{non-cooperation benchmark} (with pure information market only), where the \db~does not display the \lh's licensed \ch~information. The brown dash-dot line (with mark $\bullet$) denotes the case where the \lh~sells channels on \emph{a third-party spectrum market platform}, where the \lh~can display his licensed information for free. In a third-party scheme, the licensee displays his licensed information in a third-party platform, instead of using the database's platform. Comparing with the three-stage model, such a third-party model degenerates to a two-stage model, where the database and the licensee compete with each other for selling information or channels to {\eus} in Stage I, and {\eus} decide the best subscription decisions in Stage II.


The key observations from Figure \ref{fig:performance_vs_gamma} are as follows.

\begin{OA}\label{observation:network}
\hfill 	
\begin{description}
 \item[(a)]
	 The database's profits achieved under both schemes increase with $\lambda$ (Figure \ref{fig:performance_vs_gamma}.a), while the licensee's profits achieved under both schemes decrease with $\lambda$ (Figure \ref{fig:performance_vs_gamma}.b).
 \item[(b)]
	 When the negative network externality is dominant ($\lambda$ is small), RSS is a better choice for the database; When the positive network externaltiy is dominant ($\lambda$ is large), WPS is a better choice for the database (Figure \ref{fig:performance_vs_gamma}.a).
 \item[(c)]
	 WPS is always a better choice for the licensee (Figure \ref{fig:performance_vs_gamma}.b).
 \item[(d)]
	 The proposed RSS and WPS alway outperform the non-cooperation scheme and the third-party scheme in terms of network profit. (\emph{e.g.,} the performance gain between WPS and the non-cooperation scheme is up to $87\%$).
	 Meanwhile, both schemes can achieve a network profit close to the co-ordinated benchmark (Figure \ref{fig:performance_vs_gamma}.c).
\end{description}
\end{OA}

In the following, we discuss the the intuitions behind each observation in detail.

\emph{Observation \ref{observation:network}(a):}
When the positive network externality becomes stronger, the advanced service provides a higher utility to users.
Such a higher utility drives the equilibrium price as well as the equilibrium market share up for the database, which results in the increase of the database's profit.
To keep the leasing service attractive,
the \lh~needs to drive the equilibrium retail price down, which results in the decrease of the licensee's profit.

\emph{Observation \ref{observation:network}(b):}
As the objective of the database and the licensee are partially coordinated under RSS \cite{dana2001revenue},
the price competition under WPS is more severe than  that under RSS.
When the negative network externality is dominant, the increase of the \db's market share may severely decrease the quality of advanced service. In this case, it is better for the \db~to let the \lh~take a large fraction of the market share and share the revenue of the \lh~through RSS.
On the other hand, when the positive network externality is dominant,
a larger market share can significantly increase the attractiveness of
the \db.
Hence the \db~is able to get a higher profit with the more aggressive WPS.

\emph{Observation \ref{observation:network}(c):}
As the database charges a fixed price per transaction to the licensee under WPS, the licensee can enjoy the full benefit of putting effort to maximize its own profit (\emph{e.g.,} increasing the leasing service price to serve a minority of {\eus} with high values of $\th$). Under RSS, however, some benefit from serving {\eus} would go to the database, as the licensee needs to share a fixed portion of the revenue with the database. In such case, the licensee would like to exert a higher effort under WPS than RSS, hence can achieve a higher payoff under WPS.


\emph{Observation \ref{observation:network}(d):}
Under our proposed schemes (\emph{i.e.,} RSS and WPS), the database and the licensee negotiate with each other in Stage I. The third-party scheme does not involve such negotiation process and can not exploit the cooperation benefit, hence performs worse than our proposed schemes.
Compared with the pure information market, letting some {\eus} to lease the {\lchs} alleviates the congestion (interference) of the {\uchs}. Hence, the database can provide a good quality of service to {\eus} at a higher price, which increases the network profit.
The performance gap between the two proposed schemes and coordination benchmark is caused by the \emph{imperfect} coordination of the database and the licensee.
The database and licensee cooperate  but do not completely coordinate (\emph{i.e.,} act as a single decision maker), and we refer to this gap as the \emph{non-coordination loss}.

We want to emphasize that even though the pure cooperation model can achieve the maximum network profit, the regulators such as FCC in the United States and Ofcom in UK would not allow the licensee to be a database. As shown in the FCC's ruling, the database is certified by a third-party company which does not own the spectrum. The rational for this is to prevent the monopoly in the market that may jeopardize the interest of end-users.

\begin{figure*}
	\centering
	\includegraphics[width=2.2in]{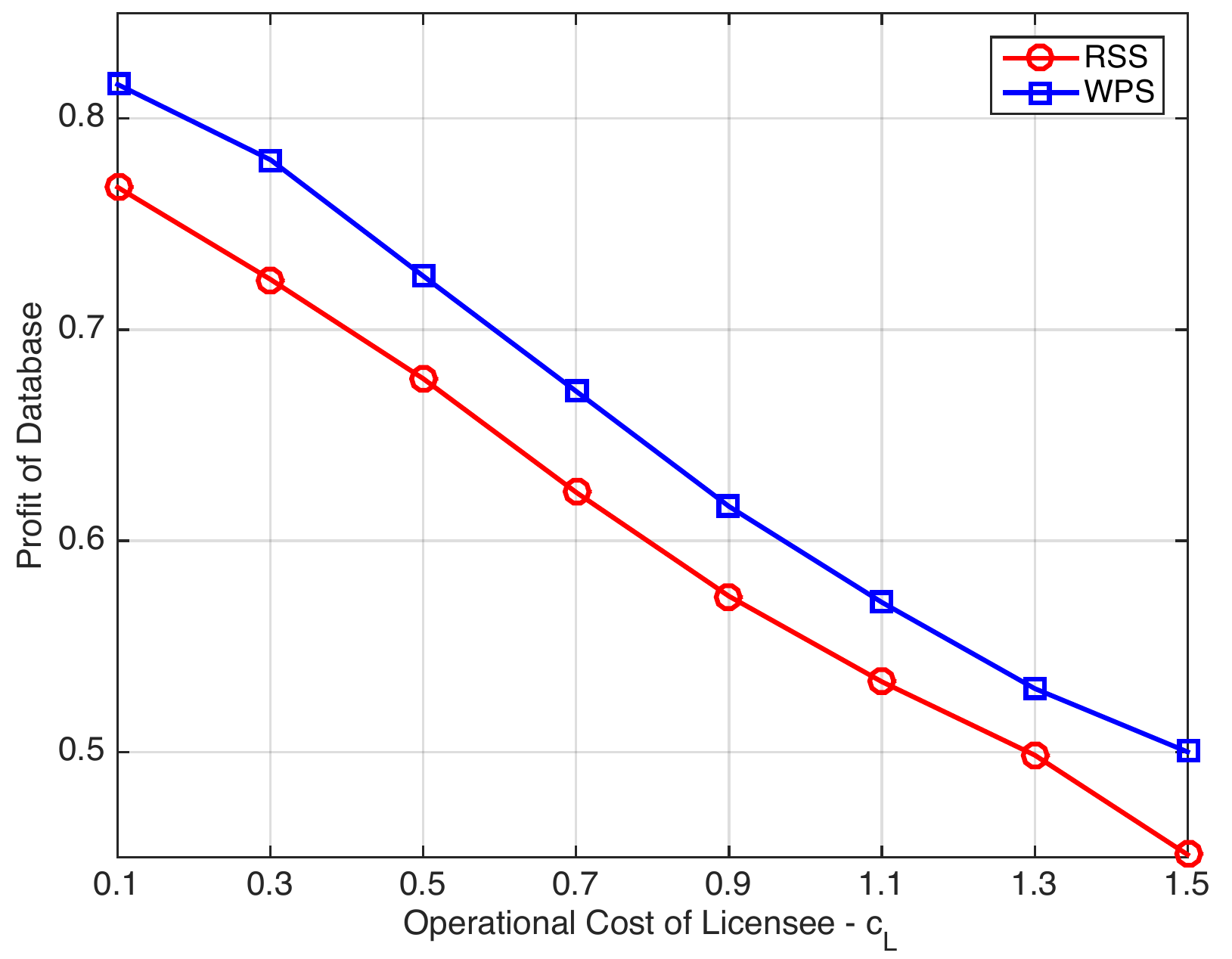}
	\includegraphics[width=2.2in]{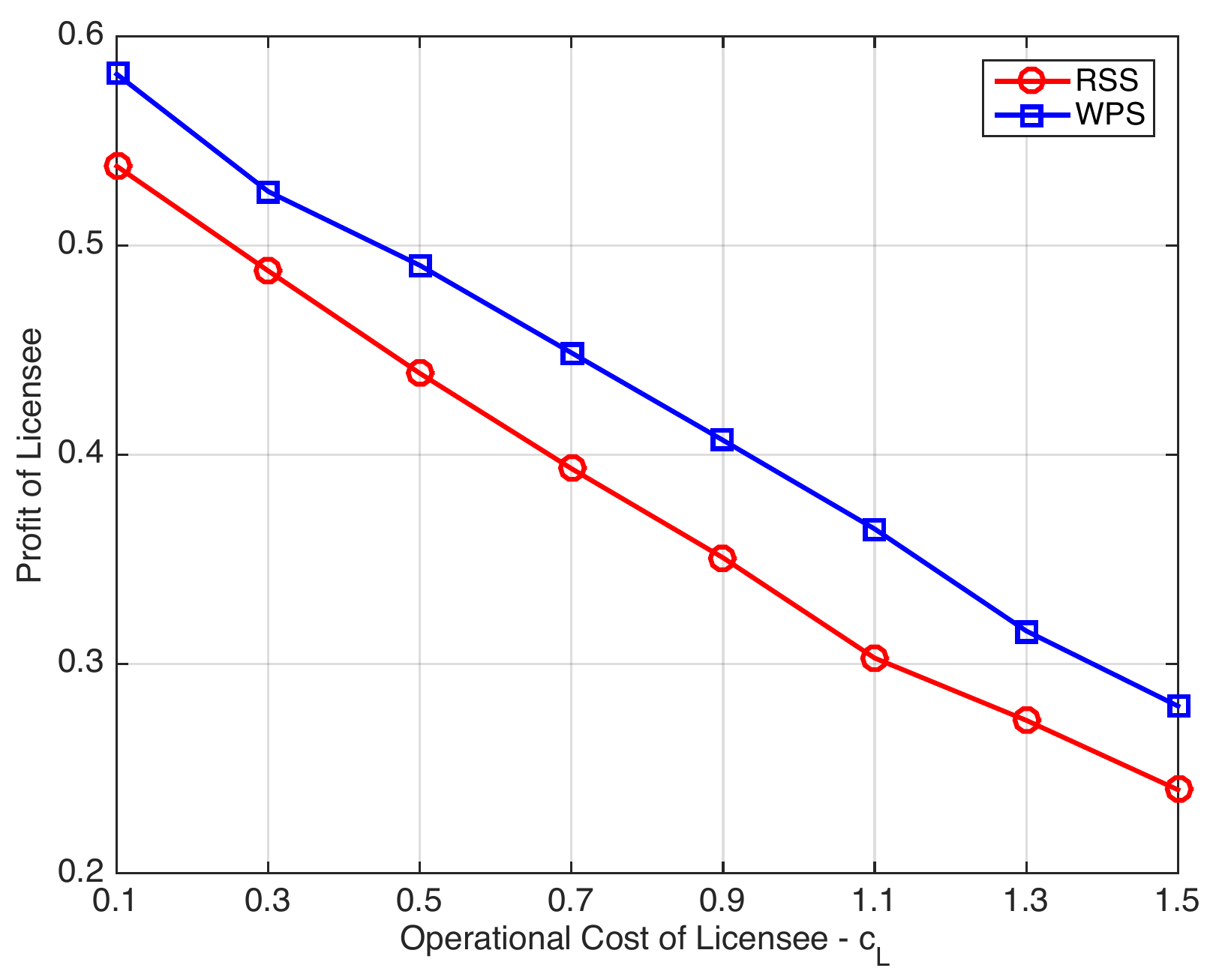}
	\includegraphics[width=2.2in]{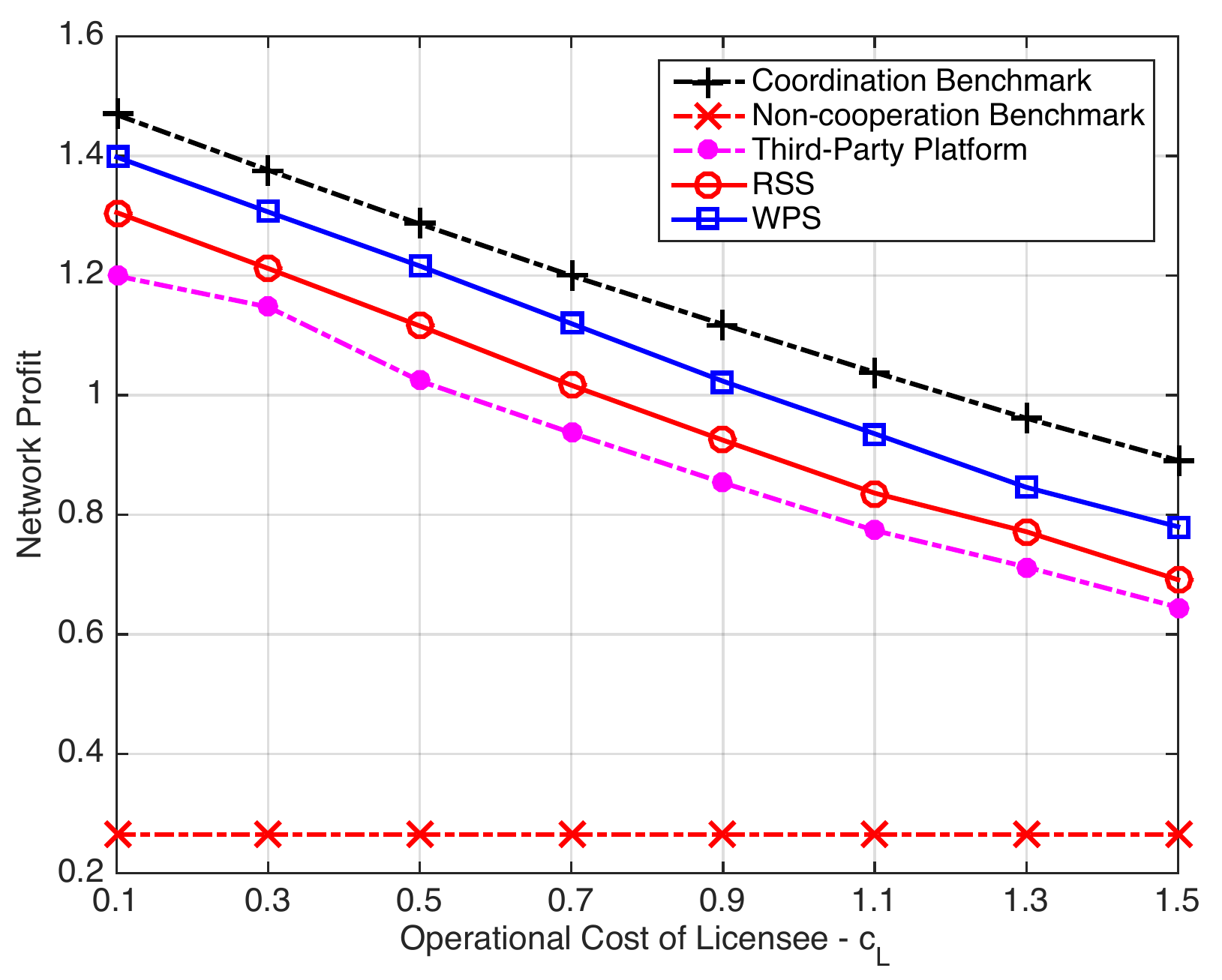}
	\caption{(a) The \db's profit, (b) The \lh's profit, (c) Network profit vs the \lh's operational cost $\cl$ under the revenue sharing scheme (RSS) and the wholesale price scheme (WPS). }\label{fig:performance_vs_ca}
\end{figure*}

\subsection{The Impact of the Licensee's Operational Cost $\cl$}
Figure \ref{fig:performance_vs_ca} is similar as Figures \ref{fig:performance_vs_gamma}, except that here we focus on the impact of licensee's operational cost $\cl$ (which incorporates costs due to energy consumptions). We choose $\cl$ from $0.1$ to $1.5$, and fix the degree of network externality as $\lambda = 1.8$ (\emph{i.e.,} the positive network externality is dominant) and the quality of leasing service as $\RL = 6$.

We list the key observations from Figure \ref{fig:performance_vs_ca} as follows.

\begin{OA}\label{observation:cost}
\hfill 	
\begin{description}
	\item[(a)]
	The database's and the licensee's profits achieved under both schemes decrease with the licensee's operational cost $\cl$ (Figures \ref{fig:performance_vs_ca}.a and  \ref{fig:performance_vs_ca}.b).
	\item[(b)]
	The non-coordination loss between the proposed two schemes and the coordination benchmark increases with the licensee's operational cost $\cl$ (Figure \ref{fig:performance_vs_ca}.c).
	\item[(c)]
	The performance gain between the proposed two schemes and the non-cooperation pure information market decreases with the licensee's operational cost $\cl$ (Figure \ref{fig:performance_vs_ca}.c).
\end{description}

\end{OA}
In the following, we discuss the intuitions of each observation in detail.

\emph{Observation \ref{observation:cost}(a):}
As the licensee's operational cost increases, the \lh~needs to increase its service price, which results in the decrease of the \lh's market share and the profit.
Meanwhile,
the \db's revenue earned from the \lh~decreases, and thus leads to the decrease of the \db's profit.

\emph{Observation \ref{observation:cost}(b):}
Due to the increase of the \lh's operational cost, the \lh~may not be willing to provide the leasing service. Hence, it is more difficult to reach an agreement during the negotiation process in Stage I, which results in a high non-coordination loss.

\emph{Observation \ref{observation:cost}(c):}
As the \lh~would raise the service price in order to compensate its revenue loss due to its higher operational cost, the leasing service becomes less attractive to {\eus}. With the decreased benefit brought by the spectrum market, the proposed schemes becoming increasingly similar as the pure information market.

\begin{figure}
	\centering
	\includegraphics[width=3.4in]{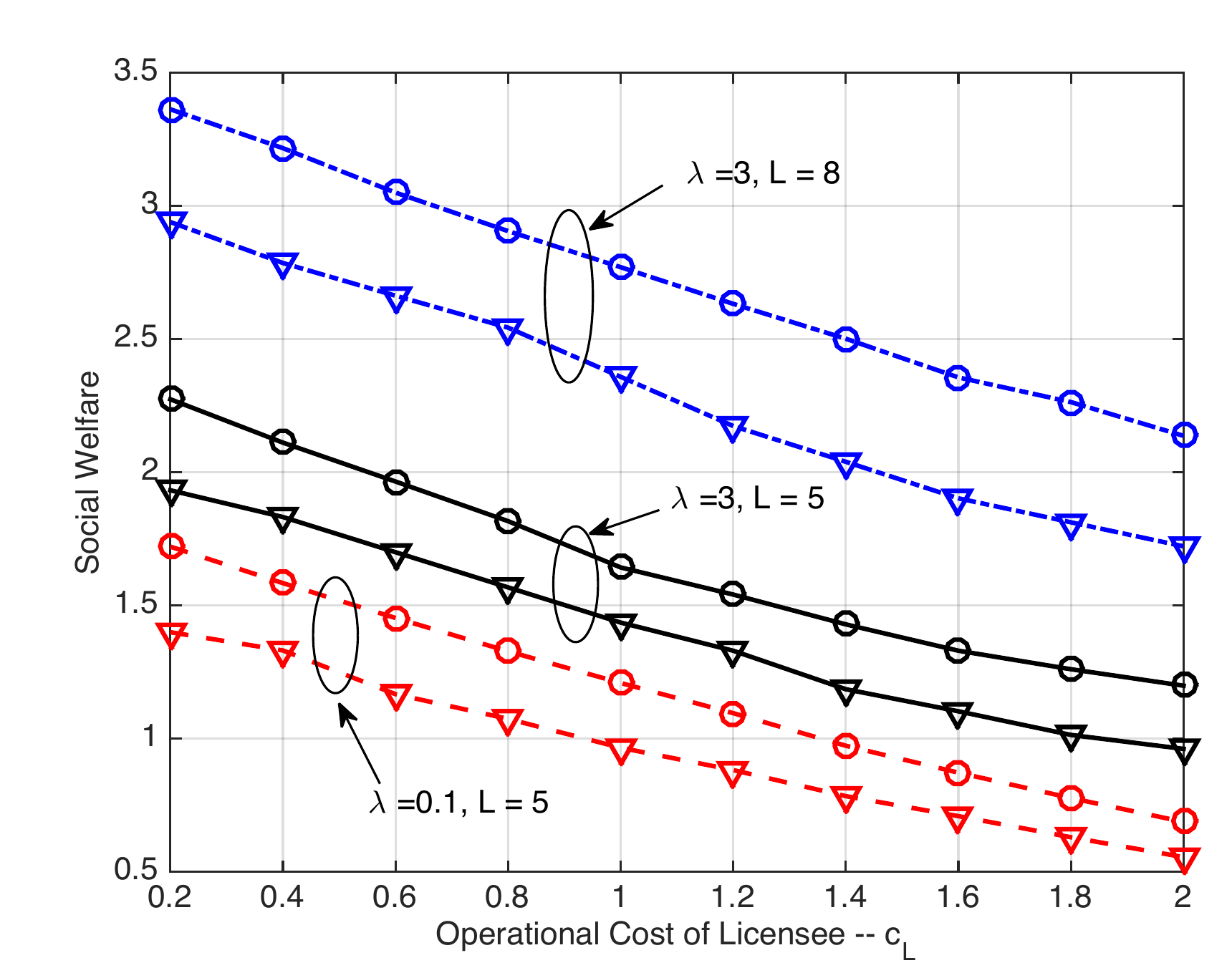}
	\caption{Social Welfare vs $\cl$ under both revenue sharing scheme (RSS) and wholesale price scheme (WPS).}\label{fig:SW-vs-diff-cl}
\end{figure}


\subsection{Social Welfare}
Figure \ref{fig:SW-vs-diff-cl} shows the social welfare, i.e, the summation of the payoffs of the \lh, the \db, and the \eus, under both RSS and WPS.
The $x$-axis represents the \lh's operational cost $\cl$ from 0.2 to 2. The lines with circle markers denote the performance of RSS, and the lines with triangle markers denote the performance of WPS. Comparing different groups of curves, we can see the impact of the degree of network externality ($\lambda$) and the quality of leasing service ($\RL$).

\begin{OA}
	RSS always outperforms WPS in terms of social welfare.
\end{OA}

Intuitively, more \eus~are attracted by the leasing service under RSS, due to less market competition between the database and the licensee. Hence RSS leads to a higher social welfare.

We also observe that the social welfare increases with leasing spectrum quality $\RL$ and decreases with the licensee's operational cost $\cl$. This is very intuitive, as the social welfare increases when the overall system is more effective (i.e., a larger $\RL$ or with a smaller $\cl$).

\section{Conclusion}\label{sec:con}
Relying database to detect the dynamically changing radio environment
can not only  reduce the energy consumption  of the unlicensed radio devices, but also
can effectively help the devices
to balance the energy consumption and communication quality.
Hence, the database-assisted TV white space network is a promising commercial applications of green cognitive communication technology.
As the success of such a network replies on a carefully designed  business model, we propose a database-provided integrated spectrum and information market, and analyzed the interactions among the geo-location database, the licensee, and the unlicensed users systematically.
We also analyze how the negative and positive network externalities (of the information market) affect these interactions.
Our work characterizes what commission charge scheme is better in terms of maximizing database's and licensee's own profits under different degrees of network externality. Specifically, when the negative network externality is dominant, RSS is a better choice for the database, while WPS is a better choice for the licensee. When the positive network externality is dominant, WPS is a better choice for both the database and the licensee.

There are several possible directions to extend this work.
First, we have assumed that the database and the licensee have the same bargaining power.
Asymmetric bargaining power will lead to a different formulation and result.
Second, we have assumed that the representative licensee has enough {\lchs} for leasing. A more practical model is to consider the scenario with limited number of {\lchs}.
In this case, the licensee needs to increase its price  at the equilibrium, so as to reduce the excess demand to the unlicensed channels. When this happens, the database may choose to increase its information price accordingly. Hence, both prices under the equilibrium are likely to be higher than those derived in our current model.
The detailed analysis for the more general model, however, is much more challenging than the one for our current model.
This is because with licensed channel limitation, the equivalent market share competition game in Stage II is no longer a super-modular game and the current analysis no longer applies.
Third, we can further investigate  an oligopoly scenario with multiple databases and licensees. If there are multiple databases that sell information of the unlicensed TV channels and multiple licensees who have the under-utilized licensed TV channels to lease, an end-user needs to decide from whom to buy the service in order to maximize its utility. Such a user subscription model in Stage III can be characterized by the Fisher market model \cite{brainard2005compute,jain2010eisenberg}, by considering the unlicensed channel information as a product of a database and the licensed TV channel as a product of a licensee. And the equilibrium product allocation can be captured by Eisenberg-Gale convex program \cite{eisenberg1959consensus}.
However, as shown in Section III.C, the information market has the positive network externality where the value of a database's product (\emph{i.e.,} the unlicensed channels' information) increases with the number of buyers. Hence, the market equilibrium in Stage III requires new analytical method that incorporates the Fisher market model with the positive network externality.
Besides, such an oligopoly scenario involves the formulation of a much more challenging one-to-many bargaining problem in Stage I (e.g., \cite{lin2014bargaining,Yu2015bargaining}).
This poses a range of new questions:  Should the database bargain with all licensees  or just  a subset of licensees?
Should the database bargain with  licensees sequentially or concurrently? If bargaining sequentially, how should the database choose the bargaining sequence to maximize its payoff?
All these questions deserve careful study in our future work.



\appendix



\section{Appendix}\label{sec:appendix}

\subsection{Property of Information Market}
In this section, we will discuss the negative and positive network externalities in the information market with a concrete example.
We first define the advanced information as the interference level on each channel, then we characterize the information value to the users. Based on that, we can further characterize the properties of the information market.

\subsubsection{Interference Information}
\label{sec:interference_info}
For each {\eu} $n\in \Nset$, it will experience
an \emph{interference level}, if it transmits on an unlicensed TV channel $\k$,
denoted by  $\InfTot_{\n,\k}$. This reflects the aggregate interference from all other nearby devices (including TV stations
and other {\eus}) operating on this channel.
Due to the time varying nature of wireless channels as well stochastic user mobilities and activities,  we consider
interference $\InfTot_{\n,\k}$ as a random variable with the following assumption:
\begin{assumption}\label{assum:iid}
For each {\eu} $n\in \Nset$,
the
interference level $\InfTot_{\n,\k}$ experienced on each channel $\k$,
is \textbf{temporal-independence} and \textbf{frequency-independence}.
\end{assumption}

This assumption shows that
(i) the interference $\InfTot_{\n,\k}$ on channel $\k$ is independent identically distributed (i.i.d.) at different times, and (ii)
the interferences on different channels, $\InfTot_{\n,\k}, \k\in\Kset$, are also i.i.d. at the same time.
\footnote{Note that the iid assumption is a reasonable approximation of the practical scenario, where all channel quality distributions are the same but the realized instant qualities of different channels are different (e.g., \cite{chang2007optimal,jiang2009optimal}). Such i.i.d. assumptions allow us to focus on the system level behaviors and derive useful insights for the next step more detailed interference modeling.}
As we concern about a generic \eu~$n$, \textbf{we will abuse the notation and omit the {\eu} index $n$ (\emph{e.g.,} write $\InfTot_{\n,\k}$ as $\InfTot_{\k}$) whenever there is no confusion caused.}
We will use
$H_{\InfTot}(\cdot)$ and $h_{\InfTot}(\cdot)$ to denote the common cumulative distribution function (CDF) and probability distribution function (PDF) of $\InfTot_{\k}$, $\forall \k\in\Kset$.\footnotesc{In general, we will conventionally  use $H_X(\cdot)$ and $h_X(\cdot)$ to denote the CDF and PDF of a random variable $X$, respectively.}~~~~

A {\eu}'s experienced interference $\InfTot_{\k}$ on a \ch~$\k$ consists of the following three components:
\begin{enumerate}
\item
$\InfTV_{\k}$: the interference from licensed TV stations;
\item
$\InfEU_{\k,m}$: the interference from another {\eu} $m$ operating on the same channel $k$;
\item
$\InfOut_{\k}$: any other interference from outside systems.
\end{enumerate}
The total interference experienced by this user on channel $k$ is
$
\InfTot_{\k} = \InfTV_k + \InfEU_k + \InfOut_k
$, where  $\InfEU_k  \triangleq \sum_{m \in \Nkset} \InfEU_{\k,m}$ is the total interference from all other {\eus} operating on channel $k$ (denoted by $\Nkset$).
Similar to $\InfTot_{\k}$, we also assume that  $ \InfTV_k , \InfEU_k , \InfEU_{\k,m}$, and $ \InfOut_k$ are random variables with \emph{temporal-independence} (\emph{i.e.,} i.i.d. across time) and {frequency-independence} ({i.e.,} i.i.d. across frequency).
We further have the following assumption:
\begin{assumption}
The interference $\InfEU_{\k,m}$ from user $m$ operating on the channel $\k$ is \textbf{user-independence}, {i.e.,} $\InfEU_{\k,m}$ are i.i.d. across user index $m$.
\end{assumption}

It is important to note that \textbf{different {\eus} may experience different interferences $\InfTV_k$ (from TV stations), $ \InfEU_{\k,m} $ (from another \eu~$m$ operating on the same \ch), and $ \InfOut_k$ (from outside systems) on a channel $k$, as we have omitted the \eu~index $n$ for all these notations for clarity.}

Next we discuss these interferences in more details.

\begin{itemize}
\item
The \db~is able to compute the interference $\InfTV_k$ from TV stations to every {\eu} (on channel $k$), as it knows the locations and channel occupancies of all TV stations.~~~~

\item
The \db~cannot compute the interference
$\InfOut_k$ from outside systems, due to the lack of outside interference source information.
Thus, the interference $\InfOut_{\k}$ will \emph{not} be included in a database's advanced information sold to {\eus}.~~~~~~~~

\item
The ability of the \db~to compute the interference $\InfEU_{\k,m}$, depends on
 whether {\eu} $m$ subscribes to the \db's advanced service.
Specifically, if {\eu} $m$ subscribes to the advanced service, the \db~can predict use $m$'s channel selection (since the {\eu} will choose the channel with the lowest interference level indicated by the \db~in the advanced information), and thus can compute its interference to any other \eu~ (based on the locations of these two users).
However, if {\eu} $m$ only chooses the \db's basic service, the \db~cannot predict its channel selection, and thus cannot compute its interference to other \eus.~~~~
\end{itemize}

For convenience, we denote $\Nkset^{[\A]}$ as the set of {\eus} operating on \ch~$k$ and subscribing to the \db's advance service (\emph{i.e.,} those choosing the strategy $\l = \A$), and $\Nkset^{[\B]}$ as the set of {\eus} operating on channel $k$ and choosing the \db's basic service (\emph{i.e.,} those choosing the strategy $\l = \B$).
That is, $\Nkset^{[\A]} \bigcup \Nkset^{[\B]} = \Nkset $.
Then, for a particular {\eu},
\begin{itemize}
  \item its experienced interference\footnote{We assume that this user has not decided which service to choose yet, and hence it is not in the set of $\mathcal{N}_k$.} (on channel $\k$)
\textbf{known by the \db}  is
  \begin{equation}\label{eq:known_inf}
	\begin{aligned}
	\textstyle
	\InfKnown_{\k} \triangleq \InfTV_{\k} + \sum_{m \in \Nkset^{[\A]}} \InfEU_{\k,m},
	\end{aligned}
   \end{equation}
	which contains the interference from TV licensees and all {\eus} (operating on channel $k$)  subscribing to the \db's advanced service.
  \item its experienced interference (on channel $k$) \textbf{\emph{not} known by the \db} is
  \begin{equation}\label{eq:unknown_inf}
	\begin{aligned}
	\textstyle
	\InfUnknown_{\k}  \triangleq \InfOut_{\k} + \sum_{m \in \Nkset^{[\B]}} \InfEU_{\k,m},
	\end{aligned}
   \end{equation}
which contains the interference from outside systems and all {\eus} (operating on channel $k$) choosing the \db's basic service.
\end{itemize}


Obviously, both $\InfUnknown_{\k}$ and $\InfKnown_{\k}$ are also random variables with temporal- and frequency-independence.
Accordingly, the total interference on {\ch} $\k$ for a {\eu} can be written as
$\textstyle
\InfTot_{\k} = \InfKnown_{\k} + \InfUnknown_{\k}.$~~~~~~

\textbf{Since the \db~knows only $\InfKnown_{\k}$, it will provide this information (instead of the total interference $\InfTot_k$) as the {advanced service} to a subscribing {\eu}.}
It is easy to see that the more {\eus} subscribing to the \db's advanced service, the more information the \db~knows, and the more accurate the \db~information will be.

Next we can characterize the accuracy of a database's information explicitly. Note that $\Proba$ and $\Probl$ denote the fraction of {\eus} choosing the advanced service and leasing licensed spectrum, respectively. Moreover, $(1 - \Proba - \Probl)$ denotes the fraction of {\eus} choosing the basic service.
Hence, there are $(1 - \Probl) \cdot \N$ {\eus} in the network that we consider operating on the (unlicensed) \chs.
Due to Assumption \ref{assum:iid}, it is reasonable to assume that each channel $k\in\Kset$ will be occupied by an average of $\frac{\N}{\K} \cdot ( 1 - \Probl)$ {\eus}.
Then, among all $\frac{\N}{\K} \cdot ( 1 - \Probl)$ {\eus} operating on channel $k$, there are, \emph{on average}, $\frac{\N}{\K}\cdot\Proba$ {\eus}  subscribing to the \db's advanced service, and $\frac{\N}{\K}\cdot ( 1 - \Proba - \Probl )$ {\eus} choosing the \db's basic service.
That is, $| \Nkset | = \frac{\N}{\K} \cdot ( 1 - \Probl)$, $| \Nkset^{[\A]} | = \frac{\N}{\K}\cdot\Proba$, and $|\Nkset^{[\B]} | = \frac{\N}{\K}\cdot(1-\Proba-\Probl)$.\footnotesc{{Note that the above discussion is from the aspect of expectation, and in a particular time period, the realized numbers of {\eus} in different channels may be different.}}
Finally, by the {user-independence} of $\InfEU_{\k,m}$, we can immediately calculate the distributions of $ \InfKnown_{\k}$ and $ \InfUnknown_{\k}$ under any given market share $\Proba$ and $\Probl$ via (\ref{eq:known_inf}) and (\ref{eq:unknown_inf}).

\subsubsection{Information Value}
Now we evaluate the value of the \db's advanced information to {\eus}, which is
reflected by the {\eu}'s benefit (utility) that can be achieved from this information.
Based on this, we can characterize the properties of functions $\fx$ and $\gy$.

We first consider the expected  utility of a {\eu} when choosing the \db's basic service (\emph{i.e.,} $\l=\B$).
In this case, the {\eu} will randomly choosing a \ch~based on the information provided in the free basic service. We use $\R_0$ to denote the rate of such a user, which is a function of the number of users choosing the basic and advanced services, $1-\Probl$:
\begin{equation}\label{eq:rate-random-fixed}
\begin{aligned}
\textstyle
\R_0( 1 - \Probl)  \textstyle = \Ex_{Z} [\rt(\InfTot)] = \int_{z} \rt(z) \mathrm{d} H_{\InfTot}(z).\\
\end{aligned}
\end{equation}
Here $\rt(\cdot)$ is the transmission rate function (\emph{e.g.,}
the Shannon capacity) as the function of the interference (assuming a fixed transmission power and channel between the transmitter and the intended receiver).
As shown in Section \ref{sec:interference_info}, each channel $k\in\Kset$ will be occupied by an average of $\frac{\N}{\K} \cdot ( 1 - \Probl)$ {\eus} based on the Assumption \ref{assum:iid}. Hence,
$\R_0( 1 - \Probl) $ depends only on the distribution of the total interference $\InfTot_k$, and thus depends on the fraction of {\eus} operating on \chs~(\emph{i.e.,} $1 - \Probl$). Then the expected utility provided by the basic service is
\begin{equation}\label{eq:utility-random-fixed}
\begin{aligned}
\textstyle
\RB(1 - \Probl) = \textstyle \ut\bigg(\R_0( 1 - \Probl)\bigg),
\end{aligned}
\end{equation}
where $\ut(\cdot)$ is the utility function of the \eu. We can easily check that
the more {\eus} operating on the \chs, the higher value of $\InfTot_k$ is, and thus the lower expected utility provided by the basic service. Hence, the basic service's expected utility reflects the congestion level of the \chs.
We use the function $\fx(\cdot)$ to characterize the congestion effect and have $\fx(1 - \Probl) = \RB( 1 - \Probl)$.

Then we consider the expected utility of a {\eu} when subscribing to the \db's advance service.
In this case, the {\db} returns the interference information $\{\InfKnown_{k}\}_{k\in\Kset}$ to the {\eu} subscribing to the advanced service, together with the basic information such as the available channel list.
A rational {\eu} will always choose the channel with the minimum $\InfKnown_{k}$ (since $\{\InfUnknown_{k}\}_{k\in\Kset}$ are i.i.d.).
Let $\InfKnownMin^{[l]} =  \min\{ \InfKnown_{1},  \ldots, \InfKnown_{K} \} $ denote the minimum interference indicated by the \db's advanced information.
Then, the actual interference experienced by a {\eu} (subscribing to the \db's advanced service) can be formulated as the sum of two random variables, denoted by $\InfTotA = \InfKnownMin + \InfUnknown$. Accordingly, the {\eu}'s expected data rate under the strategy $\l = \A$ can be computed by
\begin{equation}\label{eq:rate-pay-fixed}
\begin{aligned}
\textstyle
\Ra(\Proba, \Probl) & \textstyle = \Ex_{\InfTotA} \big[ \rt \left( \InfTotA \right) \big] = \int_z \rt(z)  \mathrm{d} H_{\InfTotA}(z),
\end{aligned}
\end{equation}
where
$H_{\InfTotA}(z)$ is the CDF of $\InfTotA$. It is easy to see that $\Ra$ depends on the distributions of $\InfKnown_k$ and $\InfUnknown_k$, and thus depend on the market share $\Probl$ and $\Proba$.
Thus, we will write $\Ra$ as $\Ra(\Probl,\Proba)$. Accordingly, the advanced service's utility is:
 \begin{equation}\label{eq:utility-pay-fixed}
\begin{aligned}
\textstyle
\RA(\Probl,\Proba) & \textstyle  \triangleq \ut\bigg( \Ra( \Probl,\Proba ) \bigg)
\end{aligned}
\end{equation}

Note that the congestion effect also affects the value of $\RA$. However, compared with the utility of {\eu} choosing basic service, the benefit of a {\eu} subscribing to the \db's advanced information comes from
$\InfKnownMin^{[l]}$, \emph{i.e.,} the minimum interference indicated by the \db's advanced information. As the value of $\InfKnownMin^{[l]}$ depends on $\Proba$ only, we can get the approximation $\RA = \fx(1 - \Probl) + \gy(\Proba)$, where function $\gy(\cdot)$ characterizes the benefit brought by $\InfKnownMin^{[l]}$ and denotes the positive network effect.

By further checking the properties of $\RB(1 - \Probl)$ and $\RA(\Probl,\Proba)$, we have
Assumption \ref{assum:congestion} and Assumption \ref{assum:positive} in Section \ref{sec:network_externality}.





\subsection{Proof for Lemma \ref{lemma:market-share}}\label{lemma:market-share-proof}
\begin{proof}
	
	By solving (\ref{eq:utility_function}), we can get the three thresholds as:
	\begin{equation}
	\textstyle \thlb \eq \frac{ \pl}{ \RL- \RB(\Probl) },
	~~~~
	\thab \eq \frac{ \pa}{ \RA(\Proba,\Probl) - \RB(\Probl) },
	~~~~
	\thla \eq \frac{\pl-\pa}{\RL - \RA(\Proba,\Probl)}.
	\end{equation}
	

	When  $\thlb > \thab$, it is easy to check that
	$$
	\textstyle \thla - \thlb = \frac{\pl(\RA(\Proba,\Probl)-\RB(\Probl)) - \pa(\RL-\RB(\Probl))}{(\RL-\RA(\Proba,\Probl))(\RL-\RB(\Probl))} > 0 ,
	$$
	since $\RL > \RA(\Proba,\Probl) $, $\RL > \RB(\Probl) $, and
	$\pl (\RA-\RB) > \pa(\RL-\RB) $ as $\thlb > \thab$.
	Hence, we have:
	$\thla > \thlb > \thab $.
	Moreover, when  $\thlb > \thab$,
	the newly derived market shares are:
	$$\Probl = 1 -  \thla,~\Proba  = \thla - \thab,$$
	When $\thlb < \thab$, we can similarly check that
	$\thla - \thlb < 0$, and hence $\thla < \thlb < \thab.$
	Moreover,
	the newly derived market shares $\{ \Probl,\Proba \}$, given the prices $\{ \pl, \pa \}$ and initial market shares $\{\Probl^0, \Proba^0 \}$, are
	$$\Probl = 1 - \thlb, \Proba  = 0.$$
	Formally, we can get (\ref{eq:user-prob-1}).
	
	
\end{proof}

\subsection{Proof for Proposition \ref{lemma:uniqueness-eq_pt}}\label{lemma:uniqueness-eq_pt-proof}
\begin{proof}
	By Definition \ref{def:stable-pt}, $\BProb^{*} = \{ \Probl^{*}, \Proba^{*} \}$ is an equilibrium point if and only if it is a solution of (\ref{eq:market_equilibrium}).
	If $\thlb < \thab$, the solution $\BProb = \{ \Probl, \Proba \}$ should satisfy that:
	\begin{equation}\label{eq:user-profit-new1}
	\triangle \Probl (\Probl, \Proba) = 1 - \frac{\pl}{L - \fx(\Probl) } - \Probl = 0;~~\text{and}~~
	\triangle \Proba (\Probl, \Proba) = 0 - \Proba = 0.
	\end{equation}
	we can easily get that $\Proba^{*} = 0$. Moreover,
	We can check that $\left.\triangle \Probl (\Probl, \Proba)\right|_{\Probl = 0}  > 0$ and $\left.\triangle \Probl (\Probl, \Proba)\right|_{\Probl = 1}  < 0$. As $\triangle \Probl (\Probl, \Proba)$ is continuous on $[0,1]$, we can get the conclusion that $\triangle \Probl (\Probl, \Proba) = 0$ has a root on the domain of $\Probl$.
	
	If $\thlb > \thab$, the solution $\BProb^{*} = \{ \Probl^{*}, \Proba^{*} \}$ should satisfy that:
	\begin{equation}\label{eq:user-profit-new21}
	\begin{aligned}
	\triangle \Probl (\Probl, \Proba) = 1 - \frac{ \pl }{ \RL - \fx(\Probl) - \gy(\Proba)} - \Probl = 0;
	\end{aligned}
	\end{equation}
	\begin{equation}\label{eq:user-profit-new22}
	\begin{aligned}
	\triangle\Proba (\Probl, \Proba) = \frac{ \pl }{ \RL - \fx(\Probl) - \gy(\Proba)} - \frac{\pa}{ \gy(\Proba)} - \Proba = 0;
	\end{aligned}
	\end{equation}
	From the above equations, we can get:
	\begin{equation}\label{eq:user-profit-pl-vs-pa}
	\begin{aligned}
	\Probl = 1 - \Proba - \frac{\pa}{ \gy(\Proba)};
	\end{aligned}
	\end{equation}
	
	By substitution (\ref{eq:user-profit-pl-vs-pa}) into (\ref{eq:user-profit-new22}), we can transform $\triangle\Proba (\Probl, \Proba)$ with two variables $\{ \Probl, \Proba \}$ into $\triangle \hat{\Prob}_{\textsc{a}} (\Proba)$ which only has one  variable $\Proba \in [0,1]$.
	We can check that $\left.\hat{\Prob}_{\textsc{a}} (\Proba)\right|_{\Proba = 0}  > 0$ and $\left.\triangle \hat{\Prob}_{\textsc{a}} (\Proba)\right|_{\Proba = 1}  < 0$. As $\triangle \hat{\Prob}_{\textsc{a}} (\cdot)$ is continuous on $[0,1]$, we can get the conclusion that $\triangle\hat{\Prob}_{\textsc{a}} (\Proba) = 0$ has a root on the domain of $\Proba$. By plugging the solution of $\triangle\hat{\Prob}_{\textsc{a}} (\Proba^{*}) = 0$ into (\ref{eq:user-profit-pl-vs-pa}), we can get the corresponding solution of $\Probl^{*}$. Hence, we can show the existence of market equilibrium.
	%
	%
	In order to proof the uniqueness, we only need to show that the function $\triangle \Probl (\Probl, \Proba)$ in (\ref{eq:user-profit-new1}) and $\triangle \hat{\Prob}_{\textsc{a}} (\Proba)$ in (\ref{eq:user-profit-pl-vs-pa}) is strictly decreasing on $[0,1]$. We can check that if
	\begin{equation}\label{eq:user-uniquenee-cod}
	\begin{aligned}
	\gy^{\prime}(\Proba) \cdot \left[ \frac{ \pl - \pa }{ ( L - \fx(\Probl) - \gy(\Proba) )^2} + \frac{ \pa }{ \gy^2(\Proba) } \right]   \leq 1
	\end{aligned}
	\end{equation}
	then the first derivative of both $\triangle \Probl (\Probl, \Proba)$ and $\triangle \hat{\Prob}_{\textsc{a}} (\Proba)$ are non-positive. Moreover, we have
	\begin{equation}\label{eq:user-uniquenee-cod-simple}
	\begin{aligned}
	& \gy^{\prime}(\Proba) \cdot \left[ \frac{ \pl - \pa }{ ( L - \fx(\Probl) - \gy(\Proba) )^2} + \frac{ \pa }{ \gy^2(\Proba) } \right] \\
	& \leq \max{ \left( \frac{ \pl - \pa }{ ( L - \fx(\Probl) - \gy(\Proba) )} , \frac{ \pa }{ \gy^2(\Proba) } \right) \cdot \frac{ \gy^{\prime}(\Proba) }{ \gy(\Proba)  } \cdot \frac{ \RL - \RB(\Probl) }{ \RL - \RA(\Proba, \Probl) }  }
	\end{aligned}
	\end{equation}
		
Let $\kappa = \max_{(\Probl, \Proba)\in\Probset} \{ \frac{\pl - \pa}{ \RL - \RA(\Proba, \Probl)} , \frac{\pa}{\RA(\Proba, \Probl) - \RB(\Probl)} \} $,
we can get the conclusion that $\triangle \Probl (\Probl, \Proba)$ in (\ref{eq:user-profit-new1}) and $\triangle \hat{\Prob}_{\textsc{a}} (\Proba)$ in (\ref{eq:user-profit-pl-vs-pa}) is strictly decreasing on $[0,1]$ when
	\begin{equation}\label{eq:user-uniquenee-cod-simple2}
\max_{(\Probl, \Proba)\in\Probset} \frac{ \gy^{\prime}(\Proba) }{ \gy(\Proba)  } \cdot \frac{ \RL - \RB(\Probl) }{ \RL - \RA(\Proba, \Probl) }  \leq \frac{1}{\kappa},	
	\end{equation}
	
	Then we prove the convergence of market dynamics based on the contraction mapping theorem \cite{bertsekas1989parallel}.
	
Let function $h_1, h_2: [0,1] \rightarrow [0,1]$ be:
	\begin{equation}
	\label{eq:mapping_function_old}
	\begin{aligned}
	h_1(\Probl, \Proba )  =  1 - \thla(\Probl, \Proba) ~\text{and}~	h_2(\Probl, \Proba )  = \thla(\Probl, \Proba) - \thab(\Probl, \Proba)
	\end{aligned}
	\end{equation}
	
	 Let function $\boldsymbol{h}:\Probset \rightarrow \Probset$ be:
	\begin{equation}
	\label{eq:mapping_function}
	\boldsymbol{h}( \Probl, \Proba ) = ( h_1(\Probl, \Proba ), h_2(\Probl, \Proba ) ).
	\end{equation}
	
	Then we need to show that the mapping function $\boldsymbol{h}(\cdot)$ is a contraction on $\Probset$ with respect to the maximum norm if the condition \eqref{eq:stable_condition} is satisfied. Specifically, we will show that
	\begin{equation}\label{eq:mappying_maximum_norm}
	\| \boldsymbol{h}( \Probla, \Probaa ) - \boldsymbol{h}( \Problb, \Probab )\|_{\infty} \leq \kappa_d \| \BProbaa - \BProbbb \|_{\infty}
	\end{equation}
	where $\kappa_d = \kappa \max_{(\Probl, \Proba)\in\Probset} \frac{ \gy^{\prime}(\Proba) }{ \gy(\Proba)  } \cdot \frac{ \RL - \RB(\Probl) }{ \RL - \RA(\Proba, \Probl) }$, $\BProbaa = ( \Probla, \Probaa ) \in \Probset$, $\BProbbb = ( \Problb, \Probab ) \in \Probset$. We assume that $\BProbaa \geq \BProbbb$ without loss of generality.
	
	We can check that
	\begin{equation}\label{eq:mappying_maximum_norm_proof}
	\begin{aligned}
	&\| \boldsymbol{h}( \Probla, \Probaa ) - \boldsymbol{h}( \Problb, \Probab )\|_{\infty} \\
	& = \max \{ \thla(\Probla, \Probaa) - \thla(\Problb, \Probab),  \\
	& \qquad{} \qquad{} \thla(\Probla, \Probaa) - \thla(\Problb, \Probab) + \thab(\Problb, \Probab) - \thab(\Probla, \Probaa)   \} \\
	& = \thla(\Probla, \Probaa) - \thla(\Problb, \Probab) + \thab(\Problb, \Probab) - \thab(\Probla, \Probaa)\\
	& = h_2(\Probla, \Probaa ) - h_2(\Problb, \Probab ).
	\end{aligned}	
	\end{equation}
	
Let $\BProbcc = ( \Problc, \Probac )$, where $\Probac \in (\Probab, \Probaa)$, and $\Problc \in ( \Problb, \Probla)$. By mean-value theorem, we have
	\begin{equation}\label{eq:mappying_maximum_norm_proof_mean_value}
	\begin{aligned}
h_2(\Probla, \Probaa ) - h_2(\Problb, \Probab ) &= \left. \frac{ \partial{h_2}}{\partial{\Probl}}\right|_{\Probl = \Problc, \Proba = \Probac} ( \Probla - \Problb ) + \left. \frac{ \partial{h_2}}{\partial{\Proba}}\right|_{\Probl = \Problc, \Proba = \Probac} ( \Probaa - \Probab )\\
& = -\frac{\pl - \pa}{[\RL - \fx(1 - \Problc) - \gy(\Probac)]^2} \cdot \fx^{\prime}(1 - \Problc) ( \Probla - \Problb ) \\
& \quad{} + \bigg [ \frac{\pl - \pa}{[\RL - \fx(1 - \Problc) - \gy(\Probac)]^2} \cdot \gy^{\prime}(\Probac) + \frac{\pa}{[\gy(\Probac)]^2} \cdot \gy^{\prime}(\Probac) \bigg]( \Probla - \Problb )  \\
& \leq \bigg[ \frac{\pl -\pa}{[ \RL - \fx(1 - \Problc) - \gy(\Probac)  ]^2} \cdot \gy^{\prime}(\Probac) \\
& \qquad{}  + \frac{\pa}{[\gy(\Probac)]^2} \cdot \gy^{\prime}(\Probac)   \bigg] \cdot \max\{ \Probla - \Problb,  \Probla - \Problb \} \\
& \leq \kappa \bigg[  \frac{\gy^{\prime}(\Probac)}{ \RL - \fx(1 - \Problc) - \gy(\Probac) } + \frac{\gy^{\prime}(\Probac)}{\gy(\Probac)}  \bigg] \cdot \max\{ \Probla - \Problb,  \Probla - \Problb \} \\
& = \kappa \cdot  \frac{\gy^{\prime}(\Probac)}{\gy(\Probac)} \cdot \frac{\RL - \fx(1 - \Problc)}{ \RL - \fx(1 - \Problc) - \gy(\Probac) }  \cdot \max\{ \Probla - \Problb,  \Probla - \Problb \} \\
& \leq \kappa_d \| \BProbaa - \BProbbb \|_{\infty}
	\end{aligned}	
	\end{equation}

\end{proof}

\subsection{Proof for Theorem \ref{thrm:stable-eq_pt}}\label{thrm:stable-eq_pt-proof}
\begin{proof}
	By the non-decreasing property of $\gy(\cdot)$ and the non-increasing property of $\fx(\cdot)$, we have $\thlb(\Probl, \Proba) \leq \left.\thlb(\Probl, \Proba)\right|_{\Probl = 0} \leq \left.\thab(\Probl, \Proba)\right|_{\Proba = 0} \leq \thab(\Probl, \Proba)$. By directly applying Proposition \ref{lemma:uniqueness-eq_pt}, we can get the conclusion.
\end{proof}

\subsection{Proof for Proposition \ref{lemma:game_tranform}}\label{lemma:game_tranform-proof}
\begin{proof}
	If $\{\Probl^{*}, \Proba^*\}$ is a Nash equilibrium of MSCG, then we have:
	\begin{equation}\label{eq:db-share-dynamic-proof}
	\left\{
	\begin{aligned}
	\textstyle\Probl^{*} & = \arg \max_{\Probl \in [0,1]}\ \Urslrs(\Probl , \Proba^{*}),
	\\
	\textstyle\Proba^{*} & = \arg \max_{\Proba \in [0,1]}\ \Urdbrs(\Probl^{*} , \Proba).
	\end{aligned}
	\right.
	\end{equation}
	By (\ref{eq:price-market-share-rs}), we further have:
	\begin{equation}\label{eq:price-market-share-rs-proof}
	\textstyle
	\left\{
	\begin{aligned}
	\textstyle  \pl^{*}(\Probl^{*} , \Proba^{*} )   = & ( 1 - \Probl^{*} ) \cdot \left( \RL - \fx(1-\Probl^{*}) - \gy(\Proba^{*}) \right)     + ( 1 - \Probl^{*} - \Proba^{*} )\cdot \gy(\Proba^{*}) , \\
	\textstyle  \pa^{*} (\Probl^{*} , \Proba^{*} ) =& ( 1 - \Probl^{*} - \Proba^{*} )\cdot \gy(\Proba^{*}).
	\end{aligned}
	\right.
	\end{equation}
	Hence, we can easily check that $( \pl^{*} , \pa^{*} )$, where
	\begin{equation}\label{eq:db-price-dynamic-proof}
	\left\{
	\begin{aligned}
	\textstyle\pl^{*} & = \arg \max_{\pl \geq 0}\ \Uslrs(\pl , \pa^{*}),
	\\
	\textstyle\pa^{*} & = \arg \max_{\pa \geq 0}\ \Udbrs(\pl^{*} , \pa).
	\end{aligned}
	\right.
	\end{equation}
	is a Nash equilibrium of PCG.
\end{proof}

\subsection{Proof for Proposition \ref{lemma:market_share_boundary}}\label{lemma:market_share_boundary-proof}
\begin{proof}
	We first need to prove that a solution $\Probl^*$ of $\max_{\Probl \in [0,1]} \Uslrs(\Probl , \Proba)$ satisfies that $\Probl^* \in (0,1/2)$. As $\Uslrs$ is differentiable, the first-order necessary condition implies that
	\begin{equation}\label{eq:FOC_ls_proof}
	\left. \frac{ \partial{\Uslrs}  }{ \partial{\Probl} } \right|_{\Probl = \Probl^*}= (1 - 2 \Probl^*)\left( \RL - \fx(1 - \Probl^*)  \right) + (  1 - \Probl^* ) \cdot \Probl^* \cdot \fx^{'}( 1 - \Probl^* ) - \Proba \cdot \gy(\Proba) - \cl = 0,
	\end{equation}
	where $\fx^{'}(x)$ is the first-order derivative of $\fx(x)$ with respect to $x$. Based on the Assumption \ref{assum:congestion}, we have $(  1 - \Probl^* ) \cdot \Probl^* \cdot \fx^{'}( 1 - \Probl^* ) - \Proba \cdot \gy(\Proba) - \cl < 0$. Thus, $(1 - 2 \Probl^*)\left( \RL - \fx(1 - \Probl^*)  \right) > 0$. According to the Assumption \ref{assume:utility_relationship}, we further have $\Probl^* < 1/2$.
	
	Then we need to prove that a solution $\Proba^*$ of $\max_{\Proba \in [0,1]} \Udbrs(\Probl , \Proba)$ satisfies that $\Proba^* + \Probl < 1$, $\forall \Probl \in [0,1]$.
	As $\Udbrs$ is differentiable, the first-order necessary condition implies that
	\begin{equation}\label{eq:FOC_db_proof}
	\left. \frac{ \partial{\Udbrs}  }{ \partial{\Proba} } \right|_{\Proba = \Proba^*}= (1 - \Probl - \Proba^*)\left( \gy(\Proba^*) + \Proba^* \cdot \gy^{'}(\Proba^*) \right) - \Proba^* \cdot \gy(\Proba^*) - \ca - \Probl \cdot \delta \cdot \gy(\Proba^*) - \Probl \cdot \Proba^* \cdot \delta \cdot \gy^{'}(\Proba^*) = 0,
	\end{equation}
	where $\gy^{'}(x)$ is the first-order derivative of $\gy(x)$ with respect to $x$. Based on the Assumption \ref{assum:positive}, we have $- \Proba^* \cdot \gy(\Proba^*) - \ca - \Probl \cdot \delta \cdot \gy(\Proba^*) - \Probl \cdot \Proba^* \cdot \delta \cdot \gy^{'}(\Proba^*) < 0$. Thus, $(1 - \Probl - \Proba^*)\left( \gy(\Proba^*) + \Proba^* \cdot \gy^{'}(\Proba^*) \right) > 0$, which implies that  $\Proba^* + \Probl < 1$, $\forall \Probl \in [0,1]$.
\end{proof}

\subsection{Proof for Theorem \ref{thrm:NE-existence}}\label{thrm:NE-existence-proof}
\begin{proof}
	We first introduce some concepts for defining a supermodular game \cite{topkis1998supermodular}. A real $n-$diensional set $\mathcal{V}$ is a \emph{sublattice} of $\mathbb{R}^{n}$ if for any two elements $a, b \in \mathcal{V}$, the component-wise minimum, $a \vee b$, and the component-wise maximum, $a \wedge b$, are also in $\mathcal{V}$. Particularly, a compact sublattice has a smallest and largest element. A function $f(x_1, \ldots, f_N)$ has increasing differences in $(x_i, x_j)$ for all $i \neq j$ if $f(x_1, \ldots, x^{1}_i, \ldots, x_N) - f(x_1, \ldots, x^{2}_i, \ldots, x_N)$ is increasing in $x_j$ for all $x^{2}_i - x^{1}_i > 0$.\footnote{If the function $f$ is twice differentiable, the property is equivalent to $\partial^2{f}/\partial{x_i}\partial{x_j}\geq 0$} The formal definition of a supermodular game is given below:
	\begin{definition}[Supermodular Game \cite{topkis1998supermodular}]\label{def:supermodular}
		A noncooperative game $( \mathcal{M}, \{S \}_{ m \in \mathcal{M} }, \{ U_m \}_{ m \in \mathcal{M} } )$ is called a supermodular game if the following conditions are all satisfied:
		\begin{itemize}
			\item  The strategy set $ {S_m} $ is a nonempty and compact sublattice of real number.
			\item  The payoff function $U_m$ is supermodular in player $m$'s own strategy.\footnote{A function is always supermodular in a single variable.}
			\item  The payoff function $U_m$ has increasing differences in all sets of strategies.
		\end{itemize}
	\end{definition}
	
We then prove that the MSCG is a supermodular game with respect to $\Proba$ and $-\Probl$.
	Since the database and the licensee competes with a single instrument $\BProb = ( \Probl, \Proba )$ chosen from a compact set $[0,1]^2$, it suffices to show that both $\Urslrs(-\Probl , \Proba)$ and $\Urdbrs(-\Probl , \Proba)$ have increasing difference in $( -\Probl, \Proba )$. By  (\ref{eq:sl-profit-dynamic-rv-xx}), we have:
	\begin{equation}
	\begin{aligned}
	\frac{ \partial^2{ \Urslrs(-\Probl , \Proba) }}{ \partial{(-{\Prob}_{\textsc{l}})}\partial{\Proba} } = [  \gy(\Proba) + \Proba \cdot \gy^{\prime}(\Proba) ] \cdot ( 1 - \delta) \geq 0
	\end{aligned}
	\end{equation}
	\begin{equation}
	\begin{aligned}
	\frac{ \partial^2{ \Urdbrs(-\Probl , \Proba) }}{ \partial{(-{\Prob}_{\textsc{l}})}\partial{\Proba} } = [  \gy(\Proba) + \Proba \cdot \gy^{\prime}(\Proba)  ] \cdot ( 1 + \delta) \geq 0
	\end{aligned}
	\end{equation}
	where $\gy^{\prime}(\Proba)$ is the first-order derivative of $\gy(\cdot)$ with respect to $\Proba$.
	Hence we can conclude that MSCG is a supermodular game with respect to $\Proba$ and $-\Probl$.
	%
	%
	%
	
	To prove the uniqueness of NE, we only need to verify that:
	\begin{equation}
	\begin{aligned}
	- \frac{  \partial^2{ \Urslrs(-{\Prob}_{\textsc{l}} , \Proba) } }{ \partial{ (-{\Prob}_{\textsc{l}}) }^2 } &\geq \frac{  \partial^2{ \Urslrs( {\Prob}_{\textsc{l}} , \Proba)  } }{ \partial{ {(-{\Prob}_{\textsc{l}}) } }\partial{ \Proba  } }, \\
	- \frac{  \partial^2{ \Urdbrs( -{\Prob}_{\textsc{l}} , \Proba) } }{ \partial{ (-{\Prob}_{\textsc{l}}) }^2 } &\geq \frac{  \partial^2{ \Urdbrs( {\Prob}_{\textsc{l}} , \Proba)  } }{ \partial{ { \Proba } }\partial{ (-{\Prob}_{\textsc{l}})  } }.
	\end{aligned}
	\end{equation}
	The above uniqueness conditions are called dominant diagonal condition, and follow the standard supermodular game theory.
	Hence, as long as $\Urslrs(-{\Prob}_{\textsc{l}} , \Proba) $ and $\Urdbrs(-{\Prob}_{\textsc{l}} , \Proba)$ satisfy the above condition, we can conclude that the MSCG has a unique NE.
	so does the original PCG.
\end{proof}

\subsection{Proof for Theorem \ref{them:uniq_wps}}\label{them:uniq_wps-proof}
\begin{proof}
	With the similar analysis used in Section \ref{sec:layer2-rss}, we can transform the original price competiton game (PCG) into an equivalent market share competition game (MSCG), where the payoffs function of the licensee and the database are
	\begin{equation}\label{eq:sl-profit-dynamic-wp-xx}
	\left\{
	\begin{aligned}
	\textstyle
	\Urslwp(\Probl , \Proba) & = ( \pl (\Probl, \Proba)  - \w ) \cdot \Probl,
	\\
	\textstyle
	\Urdbwp(\Probl , \Proba) & = \pa (\Probl , \Proba) \cdot \Probl + \w \cdot \Probl.
	\end{aligned}
	\right.
	\end{equation}
	
	We first need to prove that a solution $\Probl^*$ of $\max_{\Probl \in [0,1]} \Uslwp(\Probl , \Proba)$ satisfies that $\Probl^* \in (0,1/2)$. As $\Uslwp$ is differentiable, the first-order necessary condition implies that
	\begin{equation}\label{eq:FOC_ls_wp_proof}
	\left. \frac{ \partial{\Uslwp}  }{ \partial{\Probl} } \right|_{\Probl = \Probl^*}= (1 - 2 \Probl^*)\left( \RL - \fx(1 - \Probl^*)  \right) + (  1 - \Probl^* ) \cdot \Probl^* \cdot \fx^{'}( 1 - \Probl^* ) - \Proba \cdot \gy(\Proba) - w - \cl = 0,
	\end{equation}
	where $\fx^{'}(x)$ is the first-order derivative of $\fx(x)$ with respect to $x$. Based on the Assumption \ref{assum:congestion}, we have $(  1 - \Probl^* ) \cdot \Probl^* \cdot \fx^{'}( 1 - \Probl^* ) - \Proba \cdot \gy(\Proba) - w - \cl < 0$. Thus, $(1 - 2 \Probl^*)\left( \RL - \fx(1 - \Probl^*)  \right) > 0$. According to the Assumption \ref{assume:utility_relationship}, we further have $\Probl^* < 1/2$.
	
	Then we need to prove that a solution $\Proba^*$ of $\max_{\Proba \in [0,1]} \Udbwp(\Probl , \Proba)$ satisfies that $\Proba^* + \Probl < 1$, $\forall \Probl \in [0,1]$.
	As $\Udbrs$ is differentiable, the first-order necessary condition implies that
	\begin{equation}\label{eq:FOC_db_wp_proof}
	\left. \frac{ \partial{\Udbwp}  }{ \partial{\Proba} } \right|_{\Proba = \Proba^*}= (1 - \Probl - \Proba^*)\left( \gy(\Proba^*) + \Proba^* \cdot \gy^{'}(\Proba^*) \right) - \Proba^* \cdot \gy(\Proba^*) - \ca = 0,
	\end{equation}
	where $\gy^{'}(x)$ is the first-order derivative of $\gy(x)$ with respect to $x$. As $- \Proba^* \cdot \gy(\Proba^*) - \ca < 0$, we, thus, have $(1 - \Probl - \Proba^*)\left( \gy(\Proba^*) + \Proba^* \cdot \gy^{'}(\Proba^*) \right) > 0$, which implies that  $\Proba^* + \Probl < 1$, $\forall \Probl \in [0,1]$.
	
	Now we only need to prove that the MSCG under WPS is a supermodular game with respect to $\Proba$ and $-\Probl$.
	Since the database and the licensee competes with a single instrument $\BProb = ( \Probl, \Proba )$ chosen from a compact set $[0,1]^2$, it suffices to show that both $\Urslrs(-\Probl , \Proba)$ and $\Urdbrs(-\Probl , \Proba)$ have increasing difference in $( -\Probl, \Proba )$. By  (\ref{eq:sl-profit-dynamic-wp-xx}), we have:
	\begin{equation}
	\begin{aligned}
	\frac{ \partial^2{ \Urslwp(-\Probl , \Proba) }}{ \partial{(-{\Prob}_{\textsc{l}})}\partial{\Proba} } =   \gy(\Proba) + \Proba \cdot \gy^{\prime}(\Proba)   \geq 0
	\end{aligned}
	\end{equation}
	\begin{equation}
	\begin{aligned}
	\frac{ \partial^2{ \Urdbwp(-\Probl , \Proba) }}{ \partial{(-{\Prob}_{\textsc{l}})}\partial{\Proba} } =   \gy(\Proba) + \Proba \cdot \gy^{\prime}(\Proba)    \geq 0
	\end{aligned}
	\end{equation}
	where $\gy^{\prime}(\Proba)$ is the first-order derivative of $\gy(\cdot)$ with respect to $\Proba$.
	Hence we can conclude that MSCG under WPS is a supermodular game with respect to $\Proba$ and $-\Probl$.
	
	To prove the uniqueness of NE, we only need to verify that:
	\begin{equation}
	\begin{aligned}
	- \frac{  \partial^2{ \Urslwp(-{\Prob}_{\textsc{l}} , \Proba) } }{ \partial{ (-{\Prob}_{\textsc{l}}) }^2 } &\geq \frac{  \partial^2{ \Urslwp( {\Prob}_{\textsc{l}} , \Proba)  } }{ \partial{ {(-{\Prob}_{\textsc{l}}) } }\partial{ \Proba  } }, \\
	- \frac{  \partial^2{ \Urdbwp( -{\Prob}_{\textsc{l}} , \Proba) } }{ \partial{ (-{\Prob}_{\textsc{l}}) }^2 } &\geq \frac{  \partial^2{ \Urdbwp( {\Prob}_{\textsc{l}} , \Proba)  } }{ \partial{ { \Proba } }\partial{ (-{\Prob}_{\textsc{l}})  } }.
	\end{aligned}
	\end{equation}
	The above uniqueness conditions are called dominant diagonal condition, and follow the standard supermodular game theory.
	Hence, as long as $\Urslwp(-{\Prob}_{\textsc{l}} , \Proba) $ and $\Urdbwp(-{\Prob}_{\textsc{l}} , \Proba)$ satisfy the above condition, we can conclude that the MSCG has a unique NE.
	so does the original PCG.
	
\end{proof}

\subsection{Best Response Iteration of Stage II}\label{appendix:best_response}
According to Appendix \ref{thrm:NE-existence-proof}, the MSCG is a supermodular game.
For a supermodular game,  several commonly used updating rules are guaranteed to converged to
the NE \cite{topkis1998supermodular}.
In this work, we use the following best response algorithm as in \cite{topkis1998supermodular}:
Starting with an arbitrary market share vector $\BProb^{0}$, where 0 is the round index.
At every round $t$, both the \db~and the \lh~update their market share based on their best response to its competitor's market share in the previous round $t-1$. That is,
\begin{equation}\label{eq:BR_updated_lh}
\begin{aligned}
\Probl^{(t)} = \arg \max_{  \Probl \in [0,1] } \Urslrs(\Probl , \Proba^{t-1})
\end{aligned}
\end{equation}
\begin{equation}\label{eq:BR_updated_db}
\begin{aligned}
\Proba^{(t)} = \arg \max_{  \Proba \in [0,1] } \Urdbrs(\Probl^{t-1} , \Proba)
\end{aligned}
\end{equation}
The above procedure continues until the market share equilibrium reached.
Note that if both the \db~and the \lh~use best response algorithm update starting form their smallest (largest) element of their strategy set, then the strategies monotonically converge to the smallest (largest) NE market share. As the \db~owns the online platform and can put its advanced service as high priority, we always assume that the initial market share vector as $\BProb^{0} = \{ \Probl^{0} = 0, \Proba^{0} = 1 \}$.
Formally, we show this in Algorithm \ref{alg:round_robin}.

\begin{algorithm}
	\small
	\DontPrintSemicolon
	\textbf{Initialization}\;
	~~The licensee and the database select their own initial market share $\Probl(0)$ and $\Proba(0)$, respectively;\;
	\For{each round $t=1,...,T_{\mathrm{end}}$}
	{
		\For{the licensee and the database}
		{
			The licensee chooses the market share $\Probl (t )$ at round $t$ according to (\ref{eq:BR_updated_lh}), i.e., the market share maximizing his profit under $\Proba(t-1)$;\;
			The database chooses the market share $\Proba (t )$ at round $t$ according to (\ref{eq:BR_updated_db}), i.e., the market share maximizing his profit under $\Probl(t-1)$;
		}
		\textbf{if} {$\Probl(t) == \Probl(t-1)$ and $\Proba(t) == \Proba(t-1)$} \textbf{then} break;	\qquad \emph{// convergence}
	}
	\caption{Best response update algorithm}\label{alg:round_robin}
\end{algorithm}

\subsection{Energy Saving in the Integrated Market}
In this subsection,  we will provide additional simulation that illustrate the energy saving  (when users rely on the database to identify the best channels instead of through sensing).


For convenience, we call the market proposed in our original manuscript as the \emph{integrated market}, where the database provides both the basic service and the advanced service to users, and the licensee leases its under-utilized channels to users.
In the revised manuscript, we further consider a new benchmark market called the \emph{sensing market}, where the database only provides the basic information to user without the advanced information (i.e., qualities of channels), and the licensee can still lease his under-utilized channels to users via the platform of the database.
In this sensing market, users choosing the unlicensed channels provided by the database can decide to sense all the available (unlicensed) channels in order to identify the best channel.
We compare the performance of our proposed integrated market with this benchmark sensing market in terms of the total energy cost.

More specifically, in the sensing market, we still formulate the interactions among various entities through a three-stage hierarchical model. In Stage I, the database and the spectrum  licensee negotiate the commission charge details (regarding the spectrum market platform), \emph{i.e.,} the revenue sharing factor $\delta $ under RSS, or the wholesale price $\w$ under WPS.
In Stage II, the spectrum licensee determines the price $\pl$ of the licensed channel.
In Stage III, the users determine their best choices, and dynamically update their choices based on the current market shares. The market dynamically evolves and finally reaches an equilibrium point  (if one exists).

The analysis for the sensing market is similar to that for the integrated market in Sections IV$\sim$VI. The key differences are listed below.
In Stage III of the sensing market, a type-$\th$ {\eu}'s best response is\footnote{{Here, ``iff'' stands for ``if and only if''.
		Note that we omit the case of $\th \cdot \RL -  \pl = \max\{ \th \cdot \RS(\Probl^{\t})  -  \cs ,\ \th \cdot \RB(\Probl^{\t}) \}$, $\th \cdot \RS(\Probl^{\t}) -  \cs = \max\{ \th \cdot \RL  -  \pl ,\ \th \cdot \RB(\Probl^{\t}) \}$, and $\th \cdot \RB(\Probl^{\t})  = \max\{ \th \cdot \RL - \pl, \th \cdot \RS(\Probl^{\t})  -  \cs \}$, which are negligible (\emph{i.e.,} occurring with zero probability) due to the continuous distribution of $\th$.}}~~~~
\begin{equation}\label{eq:utility_function_sensing}
\left\{
\begin{aligned}
\l_{\th}^{\dag} = \LS, & \mbox{~~~~iff~~} \th   \RL -  \pl > \max\{ \th   \RS(\Probl^{\t})  -  \cs ,\ \th   \RB(\Probl^{\t}) \},
\\
\l_{\th}^{\dag} = \textsf{s}, & \mbox{~~~~iff~~} \th   \RS(\Probl^{\t}) -  \cs >  \max\{ \th   \RL  -  \pl ,\
\th   \RB(\Probl^{\t}) \},
\\
\l_{\th}^{\dag}  = \B, & \mbox{~~~~iff~~} \th   \RB(\Probl^\t) > \max\{ \th   \RL  -  \pl, \ \th   \RS(\Probl^{\t})-  \cs \},
\end{aligned}
\right.
\end{equation}
where $\RB(\Probl^{\t}) =  \fx(1- \Probl^{\t})$ is the expected utility that a user can achieve when choosing the basic service (defined in Section III of our manuscript),
	$\RS(\Probl^{\t}) = \fx(1- \Probl^{\t}) + \gy_1$ is the expected utility that a user can achieve when choosing to sense all the available channels (where $\gy_1 $ is a constant reflecting the additional utility gain achieved from sensing instead of choosing the basic service), and $\cs$ is the user's sensing cost.
	We further assume that $\RL > \RS(\Probl) > \RA(\Proba, \Probl) > \RB(\Probl)$. The reason for $\RS(\Probl) > \RA(\Proba, \Probl)> \RB(\Probl)$ is that the user can always
	locate the best unlicensed channel under the assumption of perfect sensing. The reason for  $\RL > \RS(\Probl)$ is that there is no mutual interference and congestion on the licensed TV channels.

In Stage II of the sensing market, the licensee determines the selling price $\pl$ for its under-utilized channels, while the database does not make any pricing decision.
The payoffs of the licensee and the database in the sensing market under the RSS scheme (Scheme I) are
\begin{equation}\label{eq:u1_sensing}
\left\{
\begin{aligned}
\Usl  &= ( \pl -\cl)  \Probl   (1 - \delta)
\\
\Udb &= ( \pl - \cl)   \Probl   \delta ,
\end{aligned}
\right.
\end{equation}
and their payoffs under the WPS scheme (Scheme II) are
\begin{equation}\label{eq:u2}
\left\{
\begin{aligned}
\Usl &= (\pl -\w)  \Probl - \cl  \Probl,
\\
\Udb &= \w   \Probl.
\end{aligned}
\right.
\end{equation}
In this sensing market, the database's payoff comes from assisting licensee in leasing licensed channels.

Finally, in Stage I of the sensing market, we can use a similar method as in Section VI to  analyze the bargaining solution in the sensing market.


\begin{figure*}
	\vspace{-2mm}
	\centering
	\includegraphics[width=2.8in]{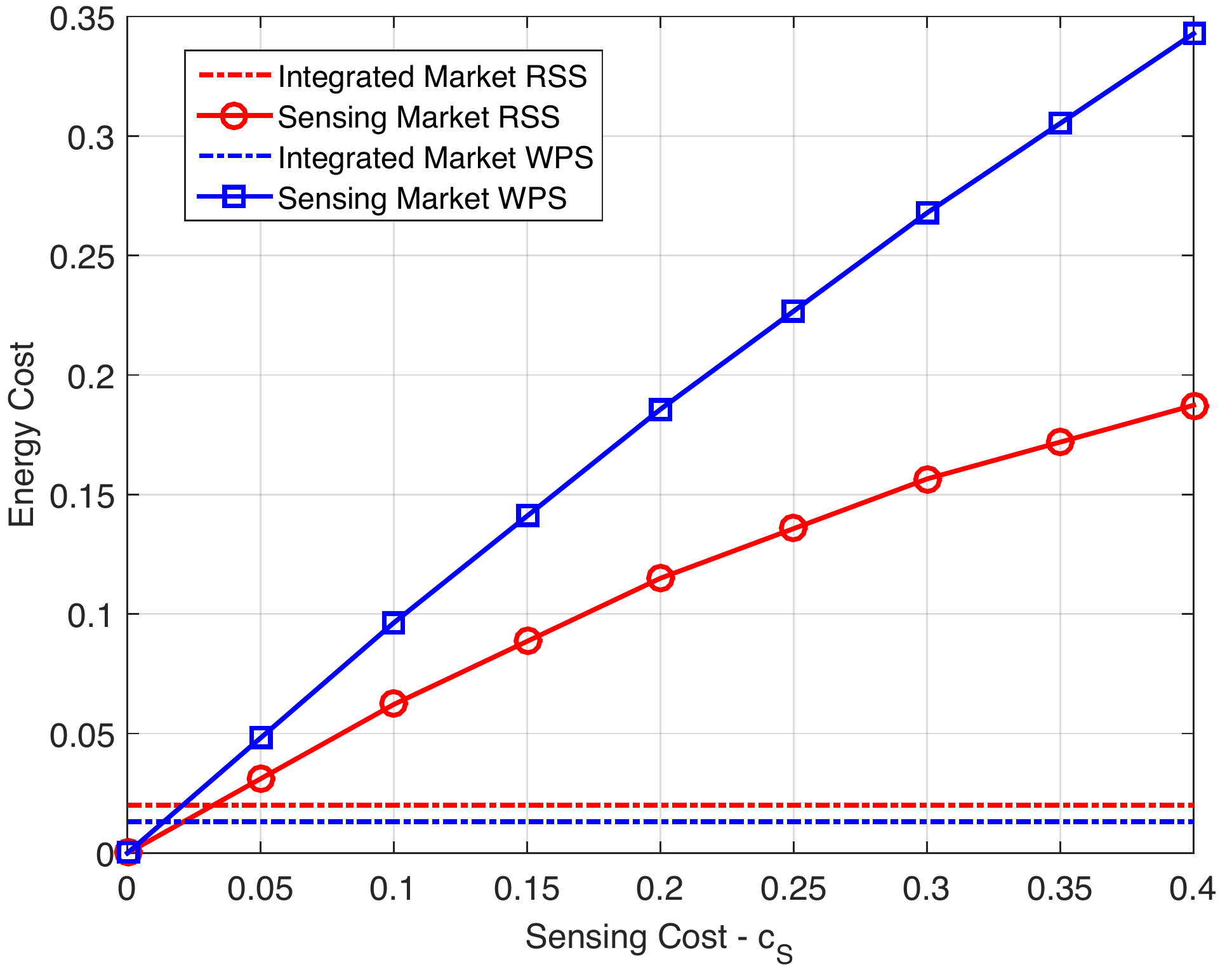}
	~~~~~~
	\includegraphics[width=2.8in]{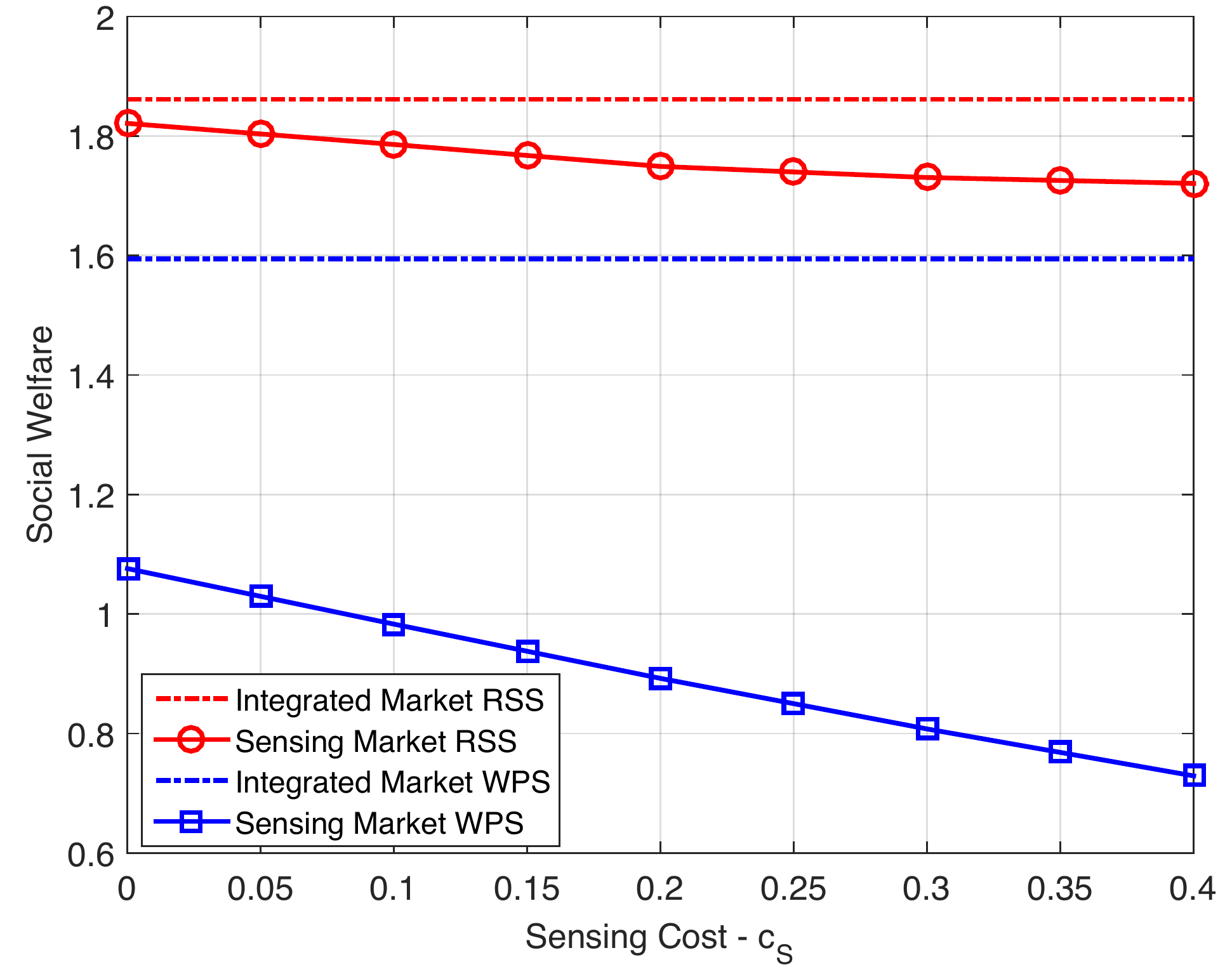}
	\vspace{-5mm}
	\caption{(a) Energy cost, and (b) Social welfare vs the sensing cost of users.}\label{fig:impact_sensing_cost}
	\vspace{-4mm}
\end{figure*}

Figure \ref{fig:impact_sensing_cost} provides the equilibrium performance comparison between the integrated market and the sensing market.
We assume that
$\alpha_1 = 1$, $\alpha_2 = 1$, $\beta_1 = 1$, $\gamma_1 = 0.6$, $\gamma_2 = 0.6$, $\ca = 0.1$,  $\cl = 0.9$, the quality of leasing service $\RL = 6$, the quality of sensing service $\RS = 2$, and $\lambda = 1.8$.

Figure \ref{fig:impact_sensing_cost}.a illustrates the energy cost achieved under different user sensing cost ($\c_S$ from $0$ to $0.4$). In this simulation, we compare the database's energy cost for providing the advanced information in the integrated market (which is calculated by $\eta_{A} \cdot \ca$) with the users' energy cost for sensing channels (which is calculated by $\eta_{S} \cdot \cs$). Here, $\eta_{S}$ denotes the fraction of users who choose the sensing service.\footnote{In the proposed integrated market, the total energy cost consists of
	(i) the users' energy cost for transmission,
	(ii) the database's energy cost for providing the available channel list (basic service), computing the channel quality information (advanced information), and assisting the licensee,
	and (iii) the licensee's energy consumption cost for leasing licensed channels to users.
	The sensing market's total energy cost is similar to that of the integrated market, except that there is no database's energy cost for providing the advanced service, and there is an extra cost for users to perform sensing.
	Hence, in this simulation, we focus on the difference between the database's energy cost for providing the advanced information and the users' sensing energy cost.}
We use the red dash-dot line to denote the energy cost in the proposed integrated market under RSS, the red line with circle markers to denote the energy cost in the sensing market under RSS, the blue dash-dot line to denote the energy cost in the proposed integrated market under WPS, and the blue line with square markers to represent the energy cost in the sensing market under WPS.

Figure \ref{fig:impact_sensing_cost}.a shows that the proposed integrated market has a smaller energy cost than the sensing market does under both RSS and WPS, except when the sensing cost is extremely small (\emph{i.e.,} $\cs \leq 0.025$).
The energy cost between the integrated market and the sensing market increases with the sensing cost. As recent studies pointed out that spectrum sensing usually introduces a high operational cost \cite{gonccalves2011value}, it is reasonable to conclude that in practice it is more energy efficient to let the database assist users in identifying the best channels.

Figure \ref{fig:impact_sensing_cost}.b illustrates the social welfare (\emph{i.e.,} the summation of the payoffs of the licensee, the database, and the users) achieved under different user sensing cost ($\c_S$ from $0$ to $0.4$).
In this figure, we use the red dash-dot line to denote the social welfare in the proposed integrated market under RSS, the red line with circle markers to denote the social welfare in the sensing market under RSS, the blue dash-dot line to denote the social welfare in the proposed integrated market under WPS, and the blue line with square markers to denote the social welfare in the sensing market under WPS.

Figure \ref{fig:impact_sensing_cost}.b shows that
the proposed integrated market outperforms the sensing market in terms of social welfare. Social welfare achieved in the sensing market decreases with the sensing cost.
Hence, it is more beneficial to let the database assist users in identifying the best channels as in our integrated market.


\begin{thebibliography}{1}


\bibitem{luo2015infocom}
Y.~Luo, L.~Gao, and J.~Huang, ``Hysim: A hybrid spectrum and information market
  for tv white space networks,'' \emph{INFOCOM}, 2015.

\bibitem{Gur2011}
G.~G$\ddot{u}$r and F.~Alag$\ddot{o}$z, ``Green wireless communications via cognitive dimension:
  an overview,'' \emph{IEEE Network}, vol.~25, no.~2, pp. 50--56, 2011.

\bibitem{climatae2008report}
The Climate Group SMART 2020 Report, ``Smart2020: Enabling the low carbon economy in the
  information age,'' [online]~\url{http://www.theclimategroup.org}.

 \bibitem{federal2012third}
Federal Communications Commission (FCC), ``Third Memorandum Opinion and Order,'' 2012.

\bibitem{Ofcom2010geo}
Ofcom, ``Implementing tv white spaces,'' 2015.

\bibitem{feng2013database}
X.~Feng, Q.~Zhang, and J.~Zhang, ``Hybrid pricing for tv white space
  database,'' \emph{IEEE INFOCOM},  2013.

\bibitem{Bogucka2012}
H.~Bogucka, M.~Parzy, P.~Marques, J.~W. Mwangoka, and T.~Forde, ``Secondary
  spectrum trading in tv white spaces,'' \emph{IEEE Communications Magazine},
  vol.~50, no.~11, pp. 121--129, 2012.

\bibitem{liu2013}
S.~Liu, H.~Zhu, R.~Du, C.~Chen, and X.~Guan, ``Location privacy preserving
  dynamic spectrum auction in cognitive radio network,'' \emph{ICDCS},
  2013.

  \bibitem{SpectrumBridgeCommericial2}
Spectrum Bridge SpecEx,
  \url{http://udia.spectrumbridge.com/ProductsServices/search.aspx}. 

\bibitem{luo2014wiopt}
Y.~Luo, L.~Gao, and J.~Huang, ``Trade information, not spectrum: A novel tv
  white space information market model,'' \emph{WiOpt}, 2014.

\bibitem{luo2014SDP}
Y.~Luo, L.~Gao, and J.~Huang, ``Information market for tv white space,'' \emph{IEEE INFOCOM Workshop on SDP}, 2014.



\bibitem{SpectrumBridgeCommericial}
Spectrum Bridge White Space Plus,
  \url{https://spectrumbridge.com/tv-white-space/white-space-plus/}.

\bibitem{harsanyi1977bargaining}
J.~C. Harsanyi, \emph{Rational behaviour and bargaining equilibrium in games
  and social situations}, Cambridge
  University Press, 1986.

\bibitem{topkis1998supermodular}
D.~M. Topkis, \emph{Supermodularity and complementarity}, Princeton University Press, 1998.

\bibitem{brandenburger2011co}
A.~Brandenburger, and J.~Barry, \emph{Co-opetition}, Crown Business, 2011.

\bibitem{hafeez2015green}
M.~Hafees and J.~Elmirghani, ``Green licensed-shared access,'' \emph{IEEE Journal on Selected Areas in Communications}, vol.~33, no.~12, pp. 2579-2595, 2015.

\bibitem{palicot2009}
J.~Palicot, ``Cognitive radio: an enabling technology for the green radio
  communications concept,'' \emph{ACM IWCMC}, 2009.

\bibitem{Ji2013}
Z.~Ji, I.~Ganchev, M.~O'Droma, and X.~Zhang, ``A realisation of broadcast
  cognitive pilot channels piggybacked on t-dmb,'' \emph{IEEE Transactions on
  Emerging Telecommunications Technologies}, vol.~24, no. 7-8, pp. 709--723,
  2013.

 \bibitem{Xie2012}
R.~Xie, F.~R.~Yu, and H.~Ji, ``Energy-Efficient Spectrum Sharing and Power Allocation in Cognitive Radio Femtocell Networks'', \emph{INFOCOM}, 2012.

\bibitem{Niyato2009game}
D.~Niyato, E.~Hossain, and Z.~Han, ``Dynamics of multiple-seller and
  multiple-buyer spectrum trading in cognitive radio networks: A game-theoretic
  modeling approach,'' \emph{IEEE Transactions on Mobile Computing}, vol.~8,
  no.~8, pp. 1009--1022, Aug 2009.

\bibitem{Min2012game}
A.~Min, X.~Zhang, J.~Choi, and K.~Shin, ``Exploiting spectrum heterogeneity in
  dynamic spectrum market,'' \emph{IEEE Transactions on Mobile Computing},
  vol.~11, no.~12, pp. 2020--2032, Dec 2012.

\bibitem{zhu2014game}
K.~Zhu, E.~Hossain, and D.~Niyato, ``Pricing, spectrum sharing, and service
  selection in two-tier small cell networks: A hierarchical dynamic game
  approach,'' \emph{IEEE Transactions on Mobile Computing}, vol.~13, no.~8,
  pp. 1843--1856, Aug 2014.

\bibitem{bhargava2004economics}
H.~K. Bhargava and V.~Choudhary, ``Economics of an information intermediary
  with aggregation benefits,'' \emph{Information Systems Research}, vol.~15,
  no.~1, pp. 22--36, 2004.

\bibitem{cachon2001contracting}
G.~P. Cachon and M.~A. Lariviere, ``Contracting to assure supply: How to share
  demand forecasts in a supply chain,'' \emph{Management Science}, vol.~47,
  no.~5, pp. 629--646, 2001.

\bibitem{gerchak2004revenue}
Y.~Gerchak and Y.~Wang, ``Revenue-sharing vs. wholesale-price contracts in
  assembly systems with random demand,'' \emph{Production and Operations
  Management}, vol.~13, no.~1, pp. 23--33, 2004.

\bibitem{dana2001revenue}
J.~D. Dana~Jr and K.~E. Spier, ``Revenue sharing and vertical control in the
  video rental industry,'' \emph{The Journal of Industrial Economics}, vol.~49,
  no.~3, pp. 223--245, 2001.

\bibitem{lariviere2001selling}
M.~A. Lariviere and E.~L. Porteus, ``Selling to the newsvendor: An analysis of
  price-only contracts,'' \emph{Manufacturing \& Service Operations
  Management}, vol.~3, no.~4, pp. 293--305, 2001.

\bibitem{niyato2008spectrum}
D.~Niyato and E.~Hossain, ``Spectrum trading in cognitive radio networks: A
  market-equilibrium-based approach,'' \emph{IEEE Wireless Communications},
  vol.~15, no.~6, pp. 71--80, 2008.

\bibitem{shetty2010congestion}
N.~Shetty, G.~Schwartz, and J.~Walrand, ``Internet qos and regulations,''
  \emph{IEEE/ACM Transactions on Networking}, vol.~18, no.~6, pp. 1725--1737,
  Dec 2010.

\bibitem{johari2010congestion}
R.~Johari, G.~Y. Weintraub, and B.~Van~Roy, ``Investment and market structure
  in industries with congestion,'' \emph{Operations Research}, vol.~58, no.~5,
  pp. 1303--1317, 2010.

\bibitem{manshaei2008evolution}
M.~Manshaei, J.~Freudiger, M.~Felegyhazi, P.~Marbach, and J.-P. Hubaux, ``On
  wireless social community networks,'' \emph{INFOCOM}, 2008.

\bibitem{shaolei2011}
S.~Ren and M.~Van~der Schaar, ``Data demand dynamics in wireless communications
  markets,'' \emph{IEEE Transactions on Signal Processing}, vol.~60, no.~4,
  pp. 1986--2000, 2012.

\bibitem{duan2011investment}
L.~Duan, J.~Huang, and B.~Shou, ``Investment and pricing with spectrum
  uncertainty: a cognitive operator's perspective,'' \emph{
  IEEE Transactions on Mobile Computing}, vol.~10, no.~11, pp. 1590--1604, 2011.

\bibitem{huang2006distributed}
J.~Huang, R.~Berry, M.~L. Honig \emph{et~al.}, ``Distributed interference
  compensation for wireless networks,'' \emph{ IEEE Journal on Selected Areas in Communications,
 }, vol.~24, no.~5, pp. 1074--1084, 2006.

\bibitem{bargai2001search}
N.~Nisan, T.~Roughgarden, E.~Tardos, and V.~V. Vazirani, \emph{Algorithmic game
  theory}, Cambridge University Press
  Cambridge, 2007.

\bibitem{easley2010effect}
D.~Easley and J.~Kleinberg, \emph{Networks, crowds, and markets},  Cambridge
  University Press, 2012.

\bibitem{wu2014exploring}
W.~Wu, R.~T. Ma, and J.~Lui, ``Exploring bundling sale strategy in online
  service markets with network effects,'' \emph{INFOCOM}, 2014.

\bibitem{chang2007optimal}
N.~B. Chang and M.~Liu, ``Optimal channel probing and transmission scheduling
  for opportunistic spectrum access,'' \emph{IEEE/ACM Transactions on Networking}, vol.~17, no.~6, pp. 1805--1818, 2009.

\bibitem{jiang2009optimal}
H.~Jiang, L.~Lai, R.~Fan, and H.~V. Poor, ``Optimal selection of channel
  sensing order in cognitive radio,'' \emph{IEEE
  Transactions on Wireless Communications}, vol.~8, no.~1, pp. 297--307, 2009.

\bibitem{bertsekas1989parallel}
D.~P.~Bertsekas and J.~N.~Tsitsiklis, \emph{Parallel and Distributed Computation: Numerical Methods}, Belmont, MA: Athena Scientific, 1989.

\bibitem{brainard2005compute}
W.~C.~Brainard and H.~E.~Scarf, ``How to Compute Equilibrium Prices in 1891," \emph{American Journal of Economics and Sociology}, vol.~64, no.~1, pp.~57-83, 2005.

\bibitem{jain2010eisenberg}
K.~Jain and V.~V.~Vazirani, ``Eisenberg-Gale Markets: Algorithms and Game-Theoretic Properties," \emph{Games and Economic Behavior}, vol.~70, no.~1, pp.~84-106, 2010.

\bibitem{eisenberg1959consensus}
E.~Eisenberg and D.~Gale, ``Consensus of Subjective Probabilities: The Pari-Mutuel Method," \emph{The Annals of Mathematical Statistics}, vol.~30, no.~1, pp.~165-168, 1959.

\bibitem{lin2014bargaining}
L.~Gao, G.~Iosifidis, J.~Huang, L.~Tassiulas, and D.~Li, ``Bargaining-based Mobile Data Offloading," \emph{IEEE Journal on Selected Areas in Communications}, vol.~32, no.~6, pp. 1114-1125, 2014.

\bibitem{Yu2015bargaining}
H.~Yu, M.~Cheung, and J.~Huang, ``Cooperative Wi-Fi Deployment: A One-to-Many Bargaining Framework," IEEE WiOpt, 2015.

\bibitem{gonccalves2011value}
V.~Goncalves and S.~Pollin, ``The Value of Sensing for TV White Spaces", IEEE DySPAN, 2011.







\end{thebibliography}
\end{document}